\documentclass[11pt,reqno]{amsart}
\usepackage{amssymb,mathrsfs,graphicx,enumerate,enumitem}
\usepackage{amsmath,amsfonts,amssymb,amscd,amsthm,bbm}
\usepackage{extpfeil}
\usepackage{graphicx,colortbl}
\usepackage{float}
\usepackage{epsfig}
\usepackage{caption}
\usepackage{subcaption}
\usepackage{bm}
\usepackage{algorithm}
\usepackage{algorithmic}
\graphicspath{{Figure/}}

\usepackage[margin=1in]{geometry}

\title[ ]{Emergent behaviors of discrete Lohe aggregation flows}

\author[Choi]{Hyungjun Choi}
\address[Hyungjun Choi]{\newline Department of Mathematical Sciences\newline Seoul National University, Seoul 08826, Republic of Korea}
\email{hana5673@snu.ac.kr}

\author[Ha]{Seung-Yeal Ha}
\address[Seung-Yeal Ha]{\newline Department of Mathematical Sciences and Research Institute of Mathematics \newline Seoul National University, Seoul 08826, Republic of Korea}
\email{syha@snu.ac.kr}

\author[Park]{Hansol Park}
\address[Hansol Park]{\newline Department of Mathematical Sciences\newline Seoul National University, Seoul 08826, Republic of Korea}
\email{hansol960612@snu.ac.kr}

\begin{document}
\newtheorem{theorem}{Theorem}[section]
\newtheorem{lemma}[theorem]{Lemma}
\newtheorem{corollary}[theorem]{Corollary}
\newtheorem{proposition}[theorem]{Proposition}
\newtheorem{remark}[theorem]{Remark}
\newtheorem{definition}[theorem]{Definition}

\renewcommand{\theequation}{\thesection.\arabic{equation}}
\renewcommand{\thetheorem}{\thesection.\arabic{theorem}}
\renewcommand{\thelemma}{\thesection.\arabic{lemma}}

\newcommand{\bx}{\mbox{\boldmath $x$}}
\newcommand{\by}{\mbox{\boldmath $y$}}
\newcommand{\bu}{\mbox{\boldmath $u$}}
\newcommand{\bv}{\mbox{\boldmath $v$}}
\newcommand{\bz}{\mbox{\boldmath $z$}}

\newcommand{\bp}{\bold p}
\newcommand{\bq}{\bold q}
\newcommand{\br}{\bold r}
\newcommand{\bbs}{\mathbb{S}}
\newcommand{\bb}{\mathbb}
\newcommand{\mbf}{\mathbf}
\newcommand{\mcal}{\mathcal}
\newcommand{\ts}{\textsuperscript}
\newcommand{\N}{\bb{N}}
\newcommand{\Z}{\bb{Z}}
\newcommand{\Q}{\bb{Q}}
\newcommand{\R}{\bb{R}}
\newcommand{\C}{\bb{C}}
\newcommand{\del}{\partial}
\newcommand{\T}{^\intercal}
\newcommand{\abs}[1]{\left\lvert #1 \right\rvert}
\newcommand{\norm}[1]{\left\lVert #1 \right\rVert}
\newcommand{\snorm}[1]{\lVert #1 \rVert}
\newcommand{\inprod}[2]{\left\langle #1, #2 \right\rangle}
\newcommand{\tr}[1]{\mathrm{tr}\left( #1 \right)}
\renewcommand{\tilde}{\widetilde}
\renewcommand{\bar}{\overline}
\newcommand{\undtxt}[2][]{\underset{\textrm{#1}}{#2}}
\newcommand{\upindex}[1]{^{(#1)}}
\newcommand{\F}{\mathrm{F}}
\newcommand{\op}{\mathrm{op}}
\newcommand{\vDelta}{\Delta^{\mcal{V}}}
\newcommand{\tvDelta}{\tilde{\Delta}^{\mcal{V}}}


{\thanks{\textbf{Acknowledgment.} The work of S.-Y. Ha was supported by National Research Foundation of Korea (NRF-2020R1A2C3A01003881).}}

\begin{abstract}
The Lohe sphere model and the Lohe matrix model are prototype continuous aggregation models on the unit sphere and the unitary group, respectively. These models have been extensively investigated in recent literature. In this paper, we propose several discrete counterparts for the continuous Lohe type aggregation models and study their emergent behaviors using the Lyapunov function method. For  suitable discretization of the Lohe sphere model, we employ a scheme consisting of two steps. In the first step, we solve the first-order forward Euler scheme, and in the second step, we project the intermediate state onto the unit sphere. For this discrete model, we present a sufficient framework leading to the complete state aggregation in terms of system parameters and initial data. For the discretization of the Lohe matrix model, we use the Lie group integrator method, Lie-Trotter splitting method and Strang splitting method to propose three discrete models. For these models, we also provide several analytical frameworks leading to complete state aggregation and asymptotic state-locking. 
\end{abstract}

\maketitle 
\centerline{\date}

\section{Introduction} \label{sec:1}
\setcounter{equation}{0}
Emergent behaviors of many-body systems can be often observed in nature. e.g. aggregation of bacteria \cite{T-B-L, T-B}, flashing of fireflies \cite{B-B, Wi2},  schooling of fish \cite{B-T1}, synchronization of pacemaker cells \cite{Pe}, etc. For survey articles and books, we refer to \cite{A-B, A-B-F, D-B1,H-K-P-Z, P-R, St, VZ, Wi1}. Despite its ubiquitous presence, systematic studies based on mathematical models were done only a half century ago by Arthur Winfree and Yoshiki Kuramoto in their seminal papers \cite{Ku2, Wi2}. Among others, we are interested in the discretization of high-dimensional Kuramoto models such as the Lohe sphere model and the Lohe matrix model. For a smooth takeoff, we begin with the Kuramoto model. Consider an ensemble of phase oscillators, and let $\theta_i = \theta_i(t)$ be the phase of the $i$-th Kuramoto oscillator. Then, the dynamics of $\theta_i$ is governed by the following phase model \cite{Ku1, Ku2}:
\begin{equation} \label{KM}
\dot{\theta}_i = \nu_i + \frac{\kappa}{N}\sum_{j=1}^N \sin(\theta_j - \theta_i),\quad i \in [N] := \{1, \cdots, N \}.
\end{equation}

The emergent dynamics of the Kuramoto model \eqref{KM} on the unit circle has been extensively studied in literature, to name a few,  \cite{B-C-M, C-H-J-K, C-S, D-X, D-B, H-K-R, H-L-X}, and its first-order discretized model for \eqref{KM} based on the forward first-order Euler method was also addressed in \cite{Ch-H, H-K-K-Z, Sh, Z-Z} from the viewpoint of emergent dynamics. As high-dimensional generalizations of the Kuramoto model, several first-order models have been proposed on specific manifolds, to name a few, the Lohe sphere model on $d$-sphere ${\mathbb S}^d$ \cite{C-C-H, C-H5, H-K-P-R, Lo-1, J-C, Lo-2, M-T-G, Ol, T-M, Zhu}, the Lohe matrix model \cite{B-C-S, D-F-M-T, D-F-M, De, H-K-R, H-R,  Lo-0} on the unitary group and the Lohe tensor model on the space of tensors with the same rank and size \cite{H-P1, H-P2}.

To set the stage, we begin with brief two aggregation models, namely the Lohe sphere model and the Lohe matrix model. First, we consider a finite ensemble of particles lying in the unit sphere $\bbs^{d-1}$. Let ${\bx}_i = {\bx}_i(t) \in \bbs^{d-1}$ be the state of the $i$-th particle on the unit sphere. Then, the Lohe sphere model reads as follows: 
\begin{equation} \label{A-1}
    \dot{\bx}_i = \Omega_i \bx_i + \frac{\kappa}{N}\sum_{j=1}^N (\bx_j - \inprod{\bx_j}{\bx_i} \bx_i),\quad  i \in [N],
\end{equation}
where $\Omega_i$ is a $d \times d$ skew-symmetric matrix.
Second, we consider the ensemble of unitary matrices, and let $U_i = U_i(t) \in {\mathbf U}(d)$ be the state of the $i$-th matrix particle. Then, the Lohe matrix model reads as follows: 
\begin{equation}  \label{A-2}
\displaystyle \dot{U}_i U_i^\dag = -\mathrm{i}H_i + \frac{\kappa}{2N}\sum_{j=1}^N (U_j U_i^\dag - U_i U_j^\dag),\quad i \in [N],
\end{equation}
where $H_i$ is a $d\times d$ Hermitian matrix. \newline

The goal of this paper is to provide first-order discrete counterparts for the continuous models \eqref{A-1} and \eqref{A-2} with emergent dynamics, and main results can be summarized as follows. Our first results deal with the suitable first-order discretization of \eqref{A-1} and sufficient frameworks leading to the complete state aggregation. \newline

First, we use the first-order forward Euler scheme to get the intermediate state and then, we project it to the unit sphere to get the state value at next step: Let $\bx_i(n)$ be a given state of the $i$-th particle at discrete time $t = n h$, where $h = \Delta t > 0$ is the time-step. Then, the intermediate state $\tilde{\bx}_i(n+1)$ and the projected state $\bx_i(n+1)$ are determined by the following recursive relations: 
\begin{equation} \label{A-3}
\begin{cases}
\displaystyle  \tilde{\bx}_i(n+1) = \bx_i(n) +h \Omega_i \bx_i(n) + \frac{\kappa h}{N} \sum_{j=1}^N \Big(\bx_j(n) - \inprod{\bx_j(n)}{\bx_i(n)} \bx_i(n) \Big), \quad n \geq 0, \\
\displaystyle \bx_i(n+1) = \frac{\tilde{\bx}_i(n+1)}{\norm{\tilde{\bx}_i(n+1)}}, \quad i \in [N].
\end{cases}
\end{equation}
Note that $\bx_i(n)$ lies in the unit sphere by construction for all $n\geq 0$, $i\in[N]$. From now on, we call system \eqref{A-3} as the discrete Lohe sphere model. \newline

For \eqref{A-3} with homogeneous free flow ($\Omega_i = \Omega$), we provide a sufficient framework leading to complete state aggregation in which all states aggregate to the same state asymptotically. If system parameters and initial data satisfy
\[ 0<\kappa h \leq 1, \quad  \min\limits_{1\leq i,j\leq N} \inprod{\bx_i^0}{\bx_j^0} > 0, \]
and let ${\mathcal X} := \{ \bx_i \}$ be a solution to \eqref{C-6} with the initial data ${\mathcal X}^0$. Then, one has complete state aggregation (Theorem \ref{T3.1}):
\[\lim_{n\to\infty} \max_{1 \leq i, j \leq N} \norm{\bx_i(n) - \bx_j(n)} = 0. \]
Second, we provide three discrete Lohe matrix models based on the Lie group integrator method, the Lie-Trotter splitting formula and the Strang splitting formula. More precisely, the first discrete Lohe matrix model is based on the Lie group integrator method, and it reads as follows. 
\begin{equation} \label{A-4}
U_i (n+1) =  \exp\left(-\mathrm{i} H_i h + \frac{\kappa h}{2}(U_c(n) U_i^\dag (n) - U_i(n) U_c^\dag (n))\right) U_i(n),\quad\text{for all } i \in [N],   
\end{equation}
where $U_c: = \frac{1}{N} \sum_{j=1}^N U_j$. Throughout the paper, we call system \eqref{A-4} as the DLM-A model for simplicity. 

For the homogeneous zero free flow with $H_i = O$, we assume that system parameters and initial data satisfy
\[ 0<\beta := \kappa h < \beta_0 \approx 0.437864, \quad  \max_{1\leq i, j \leq N} \norm{U_i^0 - U_j^0}_\F < \sqrt{ 4 - e^{2\beta} - \frac{e^{2\beta}-1}{2\beta}},\]
here, $\beta_0$ is chosen to satisfy
\[ 4 - e^{2\beta} - \frac{e^{2\beta}-1}{2\beta}>0. \]
Then, for a solution ${\mathcal U} = \{ U_i \}$ to \eqref{A-4}, complete state aggregation emerges asymptotically (see Theorem \ref{T5.1}):
\[\lim_{n\to\infty} \max_{1\leq i, j \leq N} \norm{U_i(n) - U_j(n)}_\F = 0.\]
As a second and third discrete Lohe matrix models, we propose the following models using the Lie-Trotter splitting formula and the Strang splitting formula as follows:
\begin{equation} \label{A-5}
U_i (n+1) =  \exp\big(-\mathrm{i} H_i h \big)  \exp \left( \frac{\kappa h}{2}(U_c(n) U_i^\dag (n) - U_i(n) U_c^\dag (n))\right) U_i(n),
\end{equation}
for $i \in {[N]}$, and 
\begin{equation} \label{A-6}
\displaystyle U_i (n+1)=  \exp \left(-\frac{\mathrm{i} H_i h}{2} \right)  \exp \left(\frac{\kappa h}{2}(U_c(n) U_i^\dag (n) - U_i(n) U_c^\dag (n))\right) \exp \left(-\frac{\mathrm{i} H_i h}{2} \right) U_i(n),
\end{equation}
for $ i \in [N]$ and $n \geq 0$. \newline

Note that for zero free flows with $H_i = O$, all three models \eqref{A-4}, \eqref{A-5} and \eqref{A-6} coincide so that the last two models exhibit the complete state aggregation as well. In contrast, unlike the discrete model \eqref{A-4}, we can show existence of a positively invariant set, orbital stability and asymptotic state-locking for discrete models \eqref{A-5} and \eqref{A-6} (see Theorem \ref{T6.1}, Theorem \ref{T6.2} and Theorem \ref{T6.3}).  \newline

The rest of this paper is organized as follows. In Section \ref{sec:2}, we briefly introduce continuous Lohe sphere model and continuous Lohe matrix model, and their basic properties such as gradient flow formulation. In Section \ref{sec:3}, we present a discrete analogue of the Lohe sphere model and its emergent behaviors. In Section \ref{sec:4}, we provide three discrete models using Lie group integrator method, Lie-Trotter splitting, and Strang splitting. In Section \ref{sec:5}, we provide emergent estimates for discrete Lohe matrix models for a homogeneous ensemble with $H_i = O$. In Section \ref{sec:6}, we study emergent estimates for discrete Lohe matrix models \eqref{A-5} and \eqref{A-6} for a heterogeneous ensemble.  Finally, Section \ref{sec:7} is devoted to a brief summary of our main results. 

\vspace{0.5cm}

\noindent\textbf{Gallery of Notation}: For a vector $\bx = (x_1, \cdots, x_d)$ and a square matrix $A \in {\mathbb C}^{d\times d}$, we set
\begin{align*}
& \norm{\bx} = \norm{\bx}_2 := \sqrt{|x_1|^2 + \cdots + |x_d|^2 },\quad \| A \|_{\F} := \sqrt{ \mbox{tr}( A^{\dag} A)} =\sqrt{ \mbox{tr}( A A^{\dag})}, \quad \norm{A}_{\op} := \sup_{ x\not=0} \frac{\norm{Ax}_2}{\norm{x}_2},
\end{align*}
where $\norm{\cdot}_2$ is the $2$-norm, $\norm{\cdot}_\F$ is the Frobenius norm, and $\norm{\cdot}_{\op}$ is the matrix operator norm. For an ensemble or vector of matrices ${\mathcal X} = \{\bx_i \}$ and $(X_1, \cdots, X_N)$, we use the same notation ${\mathcal X}$ to denote ensemble or vector interchangeably, and we introduce an ensemble diameter and induced ball:
\[ {\mathcal D}({\mathcal X}) := \max_{1 \leq i, j \leq N} \| \bx_i - \bx_j \|_\F, \quad {\mathcal B}(\alpha) :=  \{ {\mathcal X}  \in (\C^{d \times d})^N:~ {\mathcal D}({\mathcal X}) <\alpha\}. \]
Moreover, we set 
\[  {\mathcal U} := \{U_1, \cdots, U_N \}, \quad  {\mathcal H} := \{H_1, \cdots, H_N \}. \]

\section{Preliminaries} \label{sec:2}
\setcounter{equation}{0}
In this section, we briefly discuss two continuous systems ``{\it the Lohe sphere model }" and ``{\it the Lohe matrix model }'' for aggregation, and review  results in relation to a gradient flow formulation and emergent dynamics.
\subsection{The Lohe matrix model} \label{sec:2.1} 
Let $U_i = U_i(t)$ be a $d\times d$ unitary matrix and $H_i$ a $d\times d$ hermitian matrix whose eigenvalues represent the natural frequencies of the $i$-th Lohe oscillator. We set $U_i^{\dag}$ to be the hermitian conjugate of $U_i$. Then the temporal evolution of $U_i$ is governed by the Cauchy problem to the Lohe matrix model: 
\begin{equation}  \label{B-1}
\begin{cases}
\displaystyle \dot{U}_i U_i^\dag = -\mathrm{i}H_i + \frac{\kappa}{2N}\sum_{j=1}^N (U_j U_i^\dag - U_i U_j^\dag),\quad t > 0, \\
\displaystyle U_i \Big|_{t = 0+} = U_i^0 \in \mbf{U}(d), \quad i \in [N],
\end{cases}
\end{equation}
where $\kappa$ is a nonnegative coupling strength. \newline

In the following proposition, we list several basic properties of \eqref{B-1} without proofs. 
\begin{proposition} \label{P2.1}
\emph{\cite{Lo-1,Lo-2}} The following assertions hold:
\begin{enumerate}[label=(\roman*)]
\item
Let ${\mathcal U}$ be a solution to the Cauchy problem \eqref{B-1}. Then, $U_i U_i^{\dagger}$ is conserved:
\[  U_i(t) U_i^{\dagger}(t) = U_i^0 U_i^{0\dagger}, \quad t > 0,~i \in [N]. \]
\item
The Cauchy problem $\eqref{B-1}$ is invariant under the right-translation by a unitary matrix in the sense that if $L\in \mbf{U}(d)$ and $V_i=U_i L$, then $V_i$ satisfies the same system with translated initial data:
\begin{equation*} \label{Lohe-t}
\begin{cases}
\displaystyle {\mathrm i} \dot{V}_i V_i^{\dagger} = H_i -\frac{{\mathrm i} \kappa}{2N}\sum_{j = 1}^{N} \left( V_i V_j^{\dagger} - V_j V_i^{\dagger} \right), \quad t > 0, \\
\displaystyle V_i \Big|_{t = 0+}  =U^0_i L, \quad i \in [N].
\end{cases}
\end{equation*}
\item
(Solution splitting property): Consider the free flow with the same hamiltonian and Lohe flow without the free flow part:
\[ \dot{F}_i F_i^\dag = -\mathrm{i}H, \quad  \dot{L}_i L_i^\dag =  \frac{\kappa}{2N}\sum_{j=1}^N (L_j L_i^\dag - L_i L_j^\dag). \]
Then, the solution operator for \eqref{B-1} can be rewritten as a composition of the solution operator for the free flow and pure Lohe flow:
\[ U_i(t) = \Big( F(t) L(t) {\mathcal U}^0 \Big)_i, \quad i \in [N]. \]
\end{enumerate}
\end{proposition}
Before we discuss the emergent dynamics of \eqref{B-1}, we recall two concepts for system \eqref{B-1} as follows. 
\begin{definition} \label{D2.1}
\emph{\cite{H-R}}
\begin{enumerate}[label=(\roman*)]
\item 
The ensemble ${\mathcal U}$  is a \textit{locked state}, if $U_i U_j^{\dag}$ is time-invariant:
 \[  U_i(t)U_j^\dag(t) =  U_i^0 (U_j^0)^\dag,  \quad t > 0,~~i, j \in [N]. \]
\item  
The ensemble ${\mathcal U}$ exhibits asymptotic state-locking, if $U_i U_j^{\dag}$ has a limit as $t \to \infty$:
\[ \exists~~\lim\limits_{t\to\infty} U_i(t) U_j^\dag(t), \quad i, j  \in [N]. \]
\end{enumerate}
\end{definition}
In the sequel, we again list the emergent dynamics of \eqref{B-1} without proofs. 
\begin{proposition}  \label{P2.2}
\emph{\cite{H-R}}
\begin{enumerate}[label=(\roman*)]
\item
(Homogeneous ensemble):~Suppose system parameters and initial data satisfy
\begin{equation}  \label{B-2}
{\mathcal  D}({\mathcal H}) = 0, \quad  \kappa > 0, \quad  {\mathcal U}^0 \in \mathcal{B}(\sqrt{2}),
\end{equation}
and let ${\mathcal U}$ be a solution to \eqref{B-1}. Then, diameter ${\mathcal D}({\mathcal U})$ converges to zero exponentially as $t\to\infty$.
\vspace{0.2cm}
\item
(Heterogeneous ensemble):~Suppose system parameters and initial data satisfy
\begin{equation} \label{B-3}
{\mathcal D}({\mathcal H}) > 0, \quad  \kappa > \kappa_e > \frac{54}{17}{\mathcal D}({\mathcal H}), \quad  {\mathcal U}^0\in\mcal{B}(\alpha), 
\end{equation}
where $\alpha$ is a positive root of the following cubic polynomial equation: 
\[ x - \frac{1}{2}x^3 = \frac{{\mathcal D}({\mathcal H})}{\kappa_e} \quad \mbox{in}~\left(\sqrt{\frac{2}{3}},\sqrt{2}\right), \]
and let ${\mathcal U}$ be a solution to \eqref{B-1}. Then,  asymptotic state-locking emerges exponentially fast.
\end{enumerate}
\end{proposition}
\begin{remark}
Note that conditions \eqref{B-2} and \eqref{B-3} impose rather restrictive conditions on initial data. In contrast, for the Kuramoto model \eqref{KM}, such restriction on initial data was replaced by a generic condition using the gradient flow formulation of \eqref{KM} (see \cite{H-K-R}). So far, whether asymptotic state-locking holds for a generic initial data in a large coupling regime or not is an open problem. 
\end{remark}

Note that $\eqref{B-1}_1$ can be rewritten as 
\begin{equation*} \label{B-3-0}
\dot{U}_i  = \Big( -\mathrm{i}H_i + \frac{\kappa}{2N}\sum_{j=1}^N (U_j U_i^\dag - U_i U_j^\dag) \Big) U_i, \quad t > 0,~~ i \in [N].
\end{equation*}
For a homogeneous ensemble with the same hermitian matrices:
\begin{equation} \label{B-3-1}
 H_i = H, \quad i \in [N], 
 \end{equation}
system \eqref{B-1} can cast as a gradient flow with an analytical potential. Thanks to the solution splitting property (iii) in Proposition \ref{P2.1}, without loss of generality, we may assume 
\begin{equation*} \label{B-3-2}
 H_i = O, \quad i  \in [N].
 \end{equation*}
In this case, system \eqref{B-1} becomes
\begin{equation} \label{B-4}
 \dot{U}_i =   \frac{\kappa}{2N}\sum_{j=1}^N \left(U_j U_i^\dagger U_i -U_i U_j^\dagger U_i \right), ~~i \in [N].
\end{equation} 
Now, we introduce an order parameter $R_m$ and a potential ${\mathcal V}_m$ for \eqref{B-4} with $H_i = O$:
\begin{equation*}
  R_m^2 :=\frac{1}{N^2}\sum_{i,j=1}^N \mbox{tr}\left(U^\dag_i U_j \right), \quad \mbox{and} \quad  \mathcal{V}_m({\mathcal U}) :=-\frac{\kappa}{2} NR_m^2.
\end{equation*}
Then, it is easy to see that the potential ${\mathcal V}_m$ is analytic, and system \eqref{B-4} can cast as a gradient system with the potential ${\mathcal V}_m$.
\begin{proposition} \label{P2.3}
The following assertions hold.
\begin{enumerate}[label=(\roman*)]
\item
System \eqref{B-4} is a gradient flow with the analytical potential ${\mathcal V}_m$:
\[ {\dot U}_i = -\left.\frac{\partial \mathcal{V}_m}{\partial U_i}\right|_{T_{U_i} \mbf{U}(d)}, \quad  t > 0, \quad i \in [N],   \]
where $\frac{\partial}{\partial U_i}|_{T\mbf{U}(d)}$ is the tangential gradient. 

\vspace{0.1cm}
\item
Let ${\mathcal U}$ be a solution to \eqref{B-1} with \eqref{B-3-1}. Then, time-dependent state $e^{{\mathrm i}Ht} U_i(t)$ converges as $t\to\infty$ for any initial data ${\mathcal U}^0$.
\end{enumerate}
\end{proposition}
\vspace{0.2cm}

Next, we briefly discuss the relation between the Lohe matrix model \eqref{B-1} and the Kuramoto model \eqref{KM}. Note that for $d = 1$, an unitary matrix $U_i$ of size $1 \times 1$ corresponds to a complex number with unit modulus. So we set 
\begin{equation} \label{New-1}
U_i = e^{-{\mathrm i} \theta_i}, \quad H_i = \nu_i \in {\mathbb R},  \quad i \in [N]. 
\end{equation}
Now we substitute the above ansatz \eqref{New-1} into \eqref{B-1} to get the Kuramoto model. 
\subsection{The Lohe sphere model} \label{sec:2.2}
Let $\bx_i = \bx_i(t)$ be the position of the $i$-th swarming particle on the unit sphere $\mathbb{S}^{d-1}$, and $\Omega_i$ is a $d \times d$ skew-symmetric matrix. Then the Lohe sphere model reads as follows:
\begin{equation} \label{B-5}
\begin{cases}
\displaystyle \dot{\bx}_i = \Omega_i \bx_i + \frac{\kappa}{N}\sum_{j=1}^N \Big(\bx_j - \inprod{\bx_j}{\bx_i} \bx_i \Big ),\quad t > 0, \\
\displaystyle \bx_i \Big|_{t = 0+} = \bx_i^0, \quad i \in [N].
\end{cases}    
\end{equation}
Note that for $d = 2$, a special case for \eqref{B-5} can be derived from \eqref{B-1}. For this, we use the parametrization of the unitary matrix $U_i$ in terms of Pauli's matrices $\{ \sigma_k \}_{k=1}^{3}$:
\[ U_i := e^{-{\mathrm i} \theta_i} \Big( {\mathrm i} \sum_{k=1}^{3} x_i^k \sigma_k + x_i^4 I_2  \Big) = e^{-{\mathrm i} \theta_i}
 \left( \begin{array}{cc}
x_i^4 + {\mathrm i} x_i^1 & x_i^2 + {\mathrm i} x_i^3  \\
-x_i^2 + {\mathrm i} x_i^3 & x_i^4 - {\mathrm i} x_i^1
 \\
\end{array} \right), \]
where $I_2$ and $\sigma_i$ are the identity matrix and Pauli matrices, respectively, defined by
\[
I_2 := \left( \begin{array}{cc}
  1 & 0 \\
  0 & 1 \\
  \end{array} \right), \quad \sigma_1 := \left( \begin{array}{cc}
  1 & 0 \\
  0 & -1 \\
  \end{array} \right), \quad \sigma_2 :=  \left( \begin{array}{cc}
  0 & -{\mathrm i} \\
  {\mathrm i} & 0 \\
  \end{array} \right), \quad \sigma_3 := \left( \begin{array}{cc}
  0 & 1 \\
  1 & 0 \\
  \end{array} \right). \]
We also expand the hamiltonian matrix $H_i$:
\[ H_i = \sum_{k=1}^{3} \omega_i^k \sigma_k + \nu_i I_2, \]
where $\omega_i = (\omega_i^1, \omega_i^2, \omega_i^3)$ is a real three-vector, and the natural frequency $\nu_i$ is associated with the $\mbf{U}(1)$ component of $U_i$. After some algebraic manipulations, we obtain $5N$ equations for the angles $\theta_i$ and the four-vectors $\bx_i$:
\begin{align}
\begin{aligned} \label{B-5-1}
||\bx_i||^2 {\dot \theta}_i &= \nu_i + \frac{\kappa}{N} \sum_{k=1}^{N} \sin (\theta_k - \theta_i) \langle \bx_i,\bx_k \rangle,~~i\in[N], \\
||\bx_i||^2 {\dot \bx}_i &= \Omega_i \bx_i + \frac{\kappa}{N} \sum_{k=1}^{N} \cos(\theta_k - \theta_i) (||\bx_i||^2 \bx_k -\langle \bx_i,\bx_k \rangle \bx_i),
\end{aligned}
\end{align}
where $\Omega_i$ is a real $4 \times 4$ skew-symmetric matrix:
\[ \Omega_i := \left( \begin{array}{cccc}
0 & -\omega_i^3 & \omega_i^2 & -\omega_i^1 \\
\omega_i^3 & 0 & -\omega_i^1 & -\omega_i^2 \\
-\omega_i^2 & \omega_i^1 & 0 & -\omega_i^3 \\
\omega_i^1 & \omega_i^2 &\omega_i^3 & 0 \\
\end{array}
\right). \]
Note that the above $\Omega_i$  skew-symmetric matrix. By taking  $\theta_i = 0$ and $\nu_i = 0$ in \eqref{B-5-1}, we formally obtain the consensus model in \eqref{B-1} with $\bx_i \in \bbs^3$:
\begin{equation*} \label{B-5-2}
 {\dot \bx}_i = \Omega_i \bx_i + \frac{\kappa}{N} \sum_{k=1}^{N} (\bx_k - \langle \bx_i,\bx_k \rangle \bx_i), \quad i \in [N].
\end{equation*}
Now, we introduce order parameter $\rho$:~For a configuration ${\mathcal X} = \{\bx_i \}$,  we set
\[  \bx_c := \frac{1}{N}\sum_{j=1}^N \bx_j,\quad \rho := \| \bx_c \|. \]

In the following proposition, we list two emergent dynamics of the Lohe sphere model for a homogeneous ensemble. 
\begin{proposition} \label{P2.4} \cite{H-K-R}  
\emph{(Homogeneous ensemble)}
Suppose system parameters satisfy 
\[ \Omega_i = \Omega, \quad i \in [N], \quad \kappa > 0, \]
and let ${\mathcal X}$ be a solution to \eqref{B-5} with the initial data ${\mathcal X}^0$. Then the following assertions hold.
\begin{enumerate}[label=(\roman*)]
\item
If 
\[  \min\limits_{1 \leq i \leq N} \inprod{\bx^0_i}{\bx^0_c} > 0, \]
then $\rho(t)$ exponentially converges to $1$.
\vspace{0.1cm}
\item
If 
\[ \min\limits_{1 \leq i,j \leq N} \langle \bx^0_i, \bx^0_j \rangle > 0, \]
then $\min\limits_{1 \leq i,j \leq N} \langle \bx_i(t), \bx_j(t) \rangle$ converges to $1$ exponentially fast.
\end{enumerate}
\end{proposition}
\begin{remark}
1. Both assertions imply that ${\mathcal X}(t)$ achieves complete state aggregation:
\[\lim_{t\to\infty} \max_{1 \leq i,j \leq N} \norm{\bx_i - \bx_j} = 0. \]
2. Note that 
\[   \inprod{\bx^0_i}{\bx^0_c} = \frac{1}{N} \sum_{j=1}^{N} \inprod{\bx^0_i}{\bx^0_j} \geq \min\limits_{1 \leq i,j \leq N} \langle \bx^0_i, \bx^0_j \rangle. \]
This yields
\[ \min\limits_{1 \leq i \leq N} \inprod{\bx^0_i}{\bx^0_c} \geq \min\limits_{1 \leq i, j \leq N} \langle \bx^0_i, \bx^0_j \rangle. \]
Therefore, the condition in the second statement is more relaxed compared to the condition of the first statement. \newline

\noindent 3. For a heterogeneous ensemble, we do not have a good theory of asymptotic state-locking up to now. 
\end{remark}
\vspace{0.2cm}

Thanks to solution splitting property for \eqref{B-5} similar to Proposition \ref{P2.1} (iii), we may assume $\Omega_i = O$ for a homogeneous ensemble. Like the Lohe matrix model, system \eqref{B-5} for a homogeneous ensemble can be rewritten as a gradient system:
\begin{equation} \label{B-6}
{\dot \bx}_i = \kappa {\mathbb P}_{\bx_i}^{\perp} \Big( \frac{1}{N} \sum_{k=1}^{N} \bx_k \Big), \quad i \in [N],
\end{equation}
where  ${\mathbb P}^{\perp}_{\bx_i}$ is the orthogonal projection onto the tangent plane perpendicular to $\bx_i$:
\[{\mathbb P}^{\perp}_{\bx_i} \by = \by - \langle \by, \bx_i \rangle \bx_i.\]
Now, we introduce a potential function ${\mathcal V}_s$:
\begin{equation*} \label{B-7}
{\mathcal V}_s({\mathcal X}) :=  - \frac{\kappa}{2N}\sum_{i,j=1}^N \langle \bx_i,\bx_j \rangle = -\frac{\kappa}{2} N \rho^2.
\end{equation*}
\begin{proposition} \label{P2.5}
Suppose system parameters satisfy
\[ \Omega_i  = \Omega, \quad i \in [N] \quad \mbox{and} \quad  \kappa > 0. \]
Then, the following assertions hold.
\begin{enumerate}[label=(\roman*)]
\item
System \eqref{B-6} with $\Omega = O$ is a gradient flow with the analytical potential ${\mathcal V}_s$:
\[ {\dot \bx}_i = -\left.\frac{\partial \mathcal{V}_s}{\partial \bx_i}\right|_{T_{\bx_i} {\mathbb S}^{d-1}}, \quad  t > 0, \quad i \in [N].   \]
\item
Let ${\mathcal X}$ be a solution to \eqref{B-6} with the initial configuration ${\mathcal X}^0$. Then, the quantities $e^{-\Omega t} \bx_i(t)$ converge as $t\to\infty$ for any initial data ${\mathcal X}^0$.
\end{enumerate}
\end{proposition}

In the following sections, we study several discrete counterparts for the Lohe sphere model and the Lohe matrix model, respectively.

\section{Discrete Lohe flow on the unit sphere} \label{sec:3}
\setcounter{equation}{0}
In this section, we first present the scheme of a discrete Lohe sphere model and then, we study its emergent dynamics under a suitable framework in terms of system parameters and initial data. 
\subsection{Euler method with projection}\label{sec:3.1}
In this subsection, we  briefly discuss the Euler method for the first-order ODE model on a general manifold. It consists of the forward first-order Euler scheme and projection onto the underlying manifold. More precisely, let $\mcal{M}$ be a one-particle state space manifold embedded in Euclidean space, and  consider the following Cauchy problem:
\begin{equation} \label{C-0}
\begin{cases}
\displaystyle \dot{\bx}_i =  {\boldmath f}_i({\mathcal X}),\quad t > 0,~~i \in [N], \\
\displaystyle \bx_i(0) = \bx_i^0\in\mcal{M}.
\end{cases}
\end{equation}
For the well-posedness of \eqref{C-0}, we require the vector field $F_i({\mathcal X})  \in  T_{\bx_i}{\mathcal M}$ so that underlying manifold ${\mathcal M}$ is positively invariant by the flow generated by \eqref{C-0}:
\begin{equation*} \label{C-1}
\bx_i^0\in\mcal{M} \quad \implies \quad \bx_i(t)\in\mcal{M},\quad\forall~t\geq 0,~~i \in [N].
\end{equation*}
Let ${\mathcal X}(n)$ be given at the $n$-th time step. Then, as a naive discretization of \eqref{C-0}, we first apply the first-order forward Euler method for  \eqref{C-0} to get an intermediate value ${\tilde {\mathcal X}}(n+1)$:
\begin{equation} \label{C-1-1}
{\tilde \bx}_i(n+1) = {\bx}_i(n) + h {\boldmath f}_i({\mathcal X}(t)), \quad i \in [N].
\end{equation}
Even if 
\[ \bx_i(n) \in {\mathcal M}, \quad {\boldmath f}_i({\mathcal X}(n)) \in T_{\bx_i(n)}\mcal{M}, \]
$\bx_i(n+1)$ given by \eqref{C-1-1} may not lie in $\mcal{M}$. Thus, the naive scheme \eqref{C-1-1} is not a dynamical system on ${\mathcal M}$ in geneal. To recast system \eqref{C-1-1} to a dynamical system on ${\mathcal M}$, the projecting step is required. \newline

In summary, the proposed discrete model for \eqref{C-0} reads as follows.
\begin{equation} \label{C-2}
\begin{cases}
\displaystyle \tilde{\bx}_i(n+1) = \bx_i(n) + h {\boldmath f}_i({\mathcal X}(n)), \quad n \geq 0, ~~i \in [N],\\
\displaystyle  \bx_i(n+1) = {\mathbb P}_{\mcal{M}} \tilde{\bx}_i(n+1), \\
\displaystyle \bx_i(0) = \bx_i^0 \in\mcal{M},
\end{cases}
\end{equation}
where ${\mathbb P}_{\mcal{M}}$ is a projection operator from ${\mathbb R}^d$ to ${\mathcal M}$.  \newline

By the construction of \eqref{C-2}, as long as the discrete model \eqref{C-2} admits a solution, one has 
\[ \bx_i(n) \in {\mathcal M}, \quad n \geq 0,~~i \in [N]. \]

\subsection{Discrete Lohe sphere model} \label{sec:3.2}

\noindent Now, we present a discrete Lohe sphere model using the discrete algorithm discussed in Section \ref{sec:3.1} and study its emergent dynamics. \newline

Note that the Lohe sphere model is the first-order consensus model on the unit sphere $\mathbb{S}^{d-1}$. The projection onto the unit sphere is given by rescaling:
\[ \mathbb{P}_{\mathbb{S}^{d-1}}(\bx) = \frac{\bx}{\norm{\bx}}, \quad  \bx \in {\mathbb R}^{d} \setminus \{\mbf{0}\}. \]
Then, the discrete Lohe sphere model with Euler's method reads as follows:
\begin{equation} \label{C-3}
\begin{cases}
\displaystyle  \tilde{\bx}_i(n+1) = \bx_i(n) +h \Omega_i \bx_i(n) + \frac{\kappa h}{N} \sum_{j=1}^N \Big(\bx_j(n) - \inprod{\bx_j(n)}{\bx_i(n)} \bx_i(n) \Big), \quad n \geq 0, \\
\displaystyle \bx_i(n+1) = \frac{\tilde{\bx}_i(n+1)}{\norm{\tilde{\bx}_i(n+1)}}, \quad i \in [N].
\end{cases}
\end{equation}
Note that $\bx_i(n)$ lies in the unit sphere for all $n\geq 0$, $i \in [N]$. \newline

It is well known that the Lohe sphere model on ${\mathbb S}^1$ can be reduced to the Kuramoto model at the continuous level. In the sequel, we see that this reduction is also valid for the discrete level as well. For this, we set 
\begin{equation} \label{C-4}
\bx_i =   \begin{pmatrix} \cos\theta_i \\ \sin\theta_i \end{pmatrix},\quad \Omega_i = \begin{pmatrix} 0 & -\nu_i \\ \nu_i & 0 \end{pmatrix}.
\end{equation}
First, it follows from $\eqref{C-3}_2$ and \eqref{C-4} that  
\begin{equation} \label{C-4-1}
\tilde{\bx}_i(n+1) = \norm{\tilde{\bx}_i(n+1)} \bx_i(n+1) = \norm{\tilde{\bx}_i(n+1)} \begin{pmatrix} \cos\theta_i(n+1) \\ \sin\theta_i(n+1)  \end{pmatrix}.
\end{equation}
On the other hand, it follows from $\eqref{C-3}_1$ that 
\begin{equation} \label{C-4-2}
\tilde{\bx}_i(n+1) 
= \begin{pmatrix} \cos\theta_i(n) \\ \sin\theta_i(n) \end{pmatrix} + \left(\nu_i h + \frac{\kappa h}{N} \sum_{j=1}^N  \sin(\theta_j(n) - \theta_i(n)) \right) \begin{pmatrix} -\sin\theta_i(n) \\ \cos\theta_i(n) \end{pmatrix}.
\end{equation}
We introduce $\phi(n) \in (-\frac{\pi}{2}, \frac{\pi}{2})$ to satisfy
\[
\tan\phi(n) =\nu_i h + \frac{\kappa h}{N} \sum_{j=1}^N  \sin(\theta_j(n) - \theta_i(n)).
\]
By combining \eqref{C-4-1} and \eqref{C-4-2}, one has 
\begin{equation} \label{C-4-3}
\norm{\tilde{\bx}_i(n+1)} \begin{pmatrix} \cos\theta_i(n+1) \\ \sin\theta_i(n+1) \end{pmatrix} = \begin{pmatrix} \cos\theta_i(n) \\ \sin\theta_i(n) \end{pmatrix} + \tan\phi(n) \begin{pmatrix} -\sin\theta_i(n) \\ \cos\theta_i(n) \end{pmatrix}.
\end{equation}{
We take squaring the relations in \eqref{C-4-3} and use trigonometric identity to find 
\begin{equation}\label{C-5}
\begin{cases}
\displaystyle \|\tilde{\bx}_i(n+1)\|=\frac{1}{\cos\phi(n) }, \quad n \geq 0, \\
\displaystyle \theta_i(n+1)= \theta_i(n) + \arctan\Bigg(\nu_i h + \frac{\kappa h}{N} \sum_{j=1}^N  \sin(\theta_j(n) - \theta_i(n)) \Bigg), \quad i \in [N].
 \end{cases}
\end{equation}
Note that for $|x| \ll 1$, one has $\arctan(x) \approx x$. So $\eqref{C-5}_2$ becomes the discrete Kuramoto model \cite{Ch-H, Sh, Z-Z}:
\begin{equation} \label{C-5-1}
\theta_i(n+1)= \theta_i(n) + \nu_i h + \frac{\kappa h}{N} \sum_{j=1}^N  \sin(\theta_j(n) - \theta_i(n)).
\end{equation}
Emergent dynamics of \eqref{C-5-1} and uniform-in-time transition to the continuous dynamics have been discussed in recent literatures, e.g., exponential synchronization \cite{Ch-H} for some restricted initial configuration, complete synchronization \cite{Sh, Z-Z} for a generic initial configuration and uniform-in-time transition from discrete dynamics to continuous dynamics \cite{H-K-K-Z}.
\subsection{Emergent dynamics} \label{sec:3.3}
Next, we return to the Lohe sphere model \eqref{B-5} for a homogeneous ensemble:
\[ \Omega_i = \Omega, \quad i \in [N]. \]
We set 
\[ \by_i(t) = e^{-\Omega t} \bx_i(t), \quad i \in [N]. \]
Then, it satisfies same equation \eqref{B-5} with every $\Omega_i$ being the zero matrix.  To sum up,  without loss of generality, we may assume
\[ \Omega_i = O, \quad i \in [N].  \]
To sum up, we consider the discrete Lohe sphere model:
\begin{equation} \label{C-6}
\begin{cases}
\displaystyle  \tilde{\bx}_i(n+1) = \bx_i(n)  + \frac{\kappa h}{N} \sum_{j=1}^N \Big(\bx_j(n) - \inprod{\bx_j(n)}{\bx_i(n)} \bx_i(n) \Big), \quad n \geq 0, \\
\displaystyle \bx_i(n+1) = \frac{\tilde{\bx}_i(n+1)}{\norm{\tilde{\bx}_i(n+1)}}, \quad i \in [N].
\end{cases}
\end{equation}
Recursive relation $\eqref{C-6}_1$ can be rewritten in a handy form using the averaged state $\bx_c:= \frac{1}{N} \sum_{j=1}^{N} \bx_j$:
\begin{equation*} \label{C-6-1}
 \tilde{\bx}_i(n+1) = \bx_i(n)  + \kappa h \Big( \bx_c(n) - \langle \bx_i(n), \bx_c(n) \rangle \bx_i(n)       \Big).
\end{equation*}
For the emergent dynamics of \eqref{C-6}, we introduce several functionals:~for $n\geq 0$ and ${\mathcal X}(n)=(\bx_1(n),\cdots, \bx_N(n)) \in\big(\mathbb{S}^d\big)^N$,
\begin{align}
\begin{aligned}  \label{C-6-1-1}
&  \rho(n) = \norm{\bx_c(n)} , \quad  \bx_c(n) := \rho(n) {\hat \bx}_{c}(n),  \quad A_i(n) := \inprod{\bx_{i}(n)}{ {\hat \bx}_c(n)}, \\
&  B_{ij}(n) :=\inprod{\bx_i(n)}{\bx_j(n)},  \quad \mcal{A}(n) :=\min\limits_{1\leq i\leq N} A_i(n), \quad  \mcal{B}(n) :=\min\limits_{1\leq i,j\leq N} B_{ij}(n).
\end{aligned}
\end{align}
To describe the behavior of the functionals $\mcal{A}$ and $\mcal{B}$, we first note that 
\begin{align}
\begin{aligned} \label{C-6-2}
& \langle {\tilde \bx}_i(n+1), {\tilde \bx}_j(n+1) \rangle  \\
& \hspace{0.5cm} = \langle \bx_i, \bx_j \rangle + \kappa h \langle \bx_i + \bx_j, \bx_c \rangle (1 - \langle \bx_i, \bx_j \rangle )  \\
& \hspace{0.7cm} + (\kappa h)^2 \Big(  \|\bx_c \|^2 - \langle \bx_i, \bx_c \rangle^2 - \langle \bx_j, \bx_c \rangle^2 + \langle \bx_i, \bx_c \rangle \cdot \langle \bx_j, \bx_c \rangle \cdot \langle \bx_i, \bx_j \rangle         \Big) \\
& \hspace{0.5cm} = \langle \bx_i, \bx_j \rangle +(\kappa h \rho) \langle \bx_i + \bx_j, {\hat \bx}_c \rangle \Big(1 - \langle \bx_i, \bx_j \rangle  \Big)  \\
& \hspace{0.7cm} + (\kappa h \rho)^2 \Big(  1 - \langle \bx_i, {\hat \bx}_c \rangle^2 - \langle \bx_j, {\hat \bx}_c \rangle^2 + \langle \bx_i, {\hat \bx}_c \rangle \cdot \langle \bx_j,  {\hat \bx}_c \rangle \cdot \langle \bx_i, \bx_j \rangle         \Big).
\end{aligned}
\end{align}
Then, \eqref{C-6-2} and $\|\bx_i(n) \|^2 = 1$ imply
\begin{equation} \label{C-6-3}
\|  \tilde{\bx}_i(n+1) \|^2 =  1 + (\kappa h\rho(n))^2 \Big( 1 - \langle {\hat \bx}_c(n), \bx_i(n) \rangle^2 \Big). 
\end{equation}
Finally, $\eqref{C-6}_2$, \eqref{C-6-2} and \eqref{C-6-3} yield
\begin{align} \label{C-7}
\begin{split}
& \inprod{\bx_i(n+1)}{\bx_j(n+1)}\\
& \hspace{0.5cm}  = \frac{
\splitfrac{\inprod{\bx_i}{\bx_j} + (\kappa h \rho) \inprod{{\hat \bx}_{c}}{\bx_i + \bx_j} (1-\inprod{\bx_i}{\bx_j}) }
{+ (\kappa h \rho)^2 \left( 1-\inprod{ {\hat \bx}_{c}}{\bx_i}^2 - \inprod{{\hat \bx}_{c}}{\bx_j}^2 +\inprod{ {\hat \bx}_{c}}{\bx_i}\inprod{{\hat \bx}_{c}}{\bx_j}\inprod{\bx_i}{\bx_j} \right)}
}
{\sqrt{1+(\kappa h \rho)^2 \left(1-\inprod{{\hat \bx}_{c}}{\bx_i}^2 \right)}\sqrt{1+(\kappa h \rho)^2 \left(1-\inprod{{\hat \bx}_{c}}{\bx_j}^2 \right)}}
\end{split}
\end{align}
for $n\geq 0,~i,j \in [N]$. Here the right hand side of \eqref{C-7} is evaluated at the $n$-th time step.

\begin{lemma} \label{L3.1}
Suppose system parameters and initial data satisfy
\[ 0< \beta := \kappa h \leq 1, \quad  \mcal{B}(0) =\min\limits_{1\leq i,j\leq N} \langle \bx_i^0, \bx_j^0 \rangle > 0, \]
and let ${\mathcal X}$ be a solution to \eqref{C-6} with the initial data ${\mathcal X}^0$. Then, the following assertions hold. 
\begin{enumerate}[label=(\roman*)]
\item For each $i, j \in [N]$,  $B_{ij}(n)$ defined in \eqref{C-6-1-1} is monotonically increasing in $n$:
\[  B_{ij}(n+1) \geq  B_{ij}(n), \quad n \geq 0.\] 

\vspace{0.1cm}

\item The functionals $\langle \bx_i, \bx_c \rangle$ and $\rho$ are monotonically increasing in $n$:
\[ \langle \bx_i(n+1), \bx_c(n+1) \rangle \geq  \langle \bx_i(n), \bx_c(n) \rangle, \quad  \rho(n+1) \geq \rho(n), \quad n \geq 0. \]
 \end{enumerate}
\end{lemma}
\begin{proof}
(i)~It suffices to show that if $\mcal{B}(n)>0$, 
\[ B_{ij}(n+1) \geq B_{ij}(n). \]
Since $\mcal{B}(n)>0$, one has
\begin{equation} \label{C-7-1}
0 < A_i(n), B_{ij}(n) \leq 1, \quad i, j \in [N].
\end{equation}
For notational simplicity, we set 
\[ \gamma(n) := \kappa h \rho(n) \leq 1. \]
We use \eqref{C-7} to see that 
\begin{align*}
\begin{aligned} \label{C-8}
& \inprod{\bx_i(n+1)}{\bx_j(n+1)} \geq \inprod{\bx_i(n)}{\bx_j(n)} = B_{ij} \\
& \hspace{1.5cm} \Longleftrightarrow \quad  B_{ij} + \gamma (A_i + A_j ) (1-B_{ij}) + \gamma^2 (1-A_i^2 - A_j^2 + A_i A_j B_{ij}) \\
& \hspace{4.5cm} \geq B_{ij} \sqrt{1+\gamma^2 (1-A_{i}^2)} \sqrt{1+\gamma^2 (1-A_j^2)} \\
& \hspace{1.5cm} \Longleftarrow \quad B_{ij} + \gamma (A_i + A_j)(1-B_{ij}) + \gamma^2 (1-A_i^2 - A_j^2 + A_i A_j B_{ij})  \\
& \hspace{4.5cm} \geq B_{ij} +\gamma^2 \left(1-\frac{A_{i}^2+A_j^2}{2} \right) B_{ij}  \\
&  \hspace{1.5cm} \Longleftrightarrow \quad (A_i + A_j)(1-B_{ij}) \geq \gamma \left(A_i^2 + A_j^2 + B_{ij} - 1 - \frac{(A_i +A_j)^2}{2} B_{ij} \right),
\end{aligned}
\end{align*}
where we used the inequality:
\begin{equation} \label{C-8-0}
	\sqrt{1+x}\sqrt{1+y} \leq 1+\frac{x+y}{2} \quad \mbox{for $x,y\geq 0$},
\end{equation}
and $A_i,  B_{ij}$ and $\gamma$ are evaluated at the $n$-th time step. On the other hand, by the Cauchy-Schwarz inequality, one has 
\begin{equation} \label{C-8-1}
    (A_i - A_j)^2 = \abs{\inprod{{\hat \bx}_{c}}{\bx_i-\bx_j}}^2 \leq \norm{\bx_i-\bx_j}^2 = 2(1-B_{ij}).
\end{equation}
Finally, we use $\gamma \leq 1$, \eqref{C-7-1} and \eqref{C-8-1} to see
\begin{align*}
\begin{aligned}
& \gamma \left(A_i^2 + A_j^2 + B_{ij} - 1 - \frac{(A_i +A_j)^2}{2} B_{ij} \right)  \leq \; A_i^2 + A_j^2 - \frac{(A_i-A_j)^2}{2} - \frac{(A_i +A_j)^2}{2} B_{ij} \\
& \hspace{1cm} = \frac{(A_i+A_j)^2}{2}(1-B_{ij}) \leq  (A_i + A_j)(1-B_{ij}).
\end{aligned}
\end{align*}

\vspace{0.2cm}

\noindent (ii)~We use the result (i) to find that for a fixed $i \in [N]$, 
\[ \langle \bx_i(n+1), \bx_j(n+1) \rangle \geq \langle \bx_i(n), \bx_j(n) \rangle, \quad n\geq 0. \]
We sum the above relation over $j \in [N]$ and then divide the resulting relation by $N$ to find 
\[ \langle \bx_i(n +1), \bx_c(n +1) \rangle \geq \langle \bx_i(n), \bx_c(n) \rangle, \quad n \geq 0. \]
On the other hand, by the definition of $\rho$ and the result of (i), one has 
\begin{align*}
\begin{aligned}
 \rho(n+1)^2 &= \langle \bx_c(n+1), \bx_c(n+1) \rangle = \frac{1}{N^2} \sum_{i,j \in [N]} \langle \bx_i(n+1), \bx_j(n+1) \rangle  \\
 &\geq  \frac{1}{N^2} \sum_{i,j \in [N]} \langle \bx_i(n), \bx_j(n) \rangle = \langle \bx_c(n), \bx_c(n) \rangle = \rho(n)^2, \quad n \geq 0.
 \end{aligned}
 \end{align*}
This yields the desired estimate.
\end{proof}
Now, we are ready to provide our first main result on the complete state aggregation of \eqref{C-6}.
\begin{theorem} \label{T3.1}
Suppose system parameters and initial data satisfy
\[ 0< \beta = \kappa h \leq 1, \quad  \mcal{B}(0) =\min\limits_{1\leq i,j\leq N} \langle \bx_i^0, \bx_j^0 \rangle > 0, \]
and let ${\mathcal X}$ be a solution to \eqref{C-6} with the initial data ${\mathcal X}^0$. Then, the complete state aggregation emerges:
\[\lim_{n\to\infty} \max_{1 \leq i, j \leq N} \norm{\bx_i(n) - \bx_j(n)} = 0. \]
\end{theorem}
\begin{proof}
Since 
\[ \| \bx_i(n) - \bx_j(n) \|^2 = 2 (1 - \langle \bx_i(n), \bx_j(n) \rangle) = 2(1- B_{ij}(n)), \]
it suffices to show 
\begin{equation} \label{C-8-2}
\lim_{n \to \infty}  B_{ij}(n) = 1. 
\end{equation}
{\it Proof of \eqref{C-8-2}}:~By Lemma \ref{L3.1}, the following quantities 
\[ \rho(n), \quad  \inprod{\bx_c(n)}{\bx_i(n)} = \rho(n) \inprod{{\hat \bx}_{c}(n)}{\bx_i(n)} = \rho(n) A_i(n),  \quad B_{ij}(n) = \inprod{\bx_i(n)}{\bx_j(n)} \]
monotonically increase, as $n$ increases and these are bounded above by $1$. Hence, there exist numbers $\rho^{\infty},~A_i^{\infty},~B_{ij}^{\infty} \in (0, 1)$ such that 
\[
\lim_{n\to\infty} \rho(n) = \rho^\infty, \quad 
\lim_{n\to\infty} \inprod{\bx_c(n)}{\bx_i(n)} = A_{i}^\infty, \quad 
\lim_{n\to\infty} \inprod{\bx_i(n)}{\bx_j(n)} = B_{ij}^\infty,
\]
for all $i,j \in [N]$. Then, one has
\[ \lim_{n\to\infty} \inprod{\bx_i(n)}{\bx_j(n)} = \lim_{n\to\infty} \inprod{\bx_i(n+1)}{\bx_j(n+1)} = B_{ij}^\infty, \]
and we use \eqref{C-8-0} to get
\begin{align}
\begin{aligned} \label{C-8-3}
& B_{ij}^\infty + \beta (A_i^\infty + A_j^\infty) (1-B_{ij}^\infty) +\beta^2 \big\{ (\rho^\infty)^2 - (A_i^\infty)^2 - (A_j^\infty)^2 + A_i^\infty A_j^\infty B_{ij}^\infty \big\} \\
& \hspace{0.5cm} = B_{ij}^\infty \sqrt{1+\beta^2 \big\{(\rho^\infty)^2-(A_{i}^\infty)^2\big\}} \sqrt{1+\beta^2 \big\{(\rho^\infty)^2-(A_j^\infty)^2\big\}}.
\end{aligned}
\end{align}
Since 
\[ (A_i(n) - A_j(n))^2 = \abs{\inprod{\bx_c}{\bx_i - \bx_j}}^2 \leq \norm{\bx_c}^2 \norm{\bx_i-\bx_j}^2 = 2\rho(n)^2 (1-B_{ij}(n)), \]
we have 
\[ (A_i^\infty - A_j^\infty)^2 \leq 2(\rho^\infty)^2 (1-B_{ij}^\infty). \]
On the other hand, note the following quantity:
\[ \Delta B_{ij}(n) = B_{ij}(n +1) - B_{ij}(n). \]
We take $n \to \infty$ to $\Delta B_{ij}(n)$ and use \eqref{C-7-1} and \eqref{C-8-3} to get 
\begin{align*}
\begin{aligned}
0 &\geq B_{ij}^\infty + \beta (A_i^\infty + A_j^\infty) (1-B_{ij}^\infty) +\beta^2 \big\{ (\rho^\infty)^2 - (A_i^\infty)^2 - (A_j^\infty)^2 + A_i^\infty A_j^\infty B_{ij}^\infty \big\} \\
& \hspace{0.2cm} - B_{ij}^\infty \left( 1+ \beta^2 \left((\rho^\infty)^2 - \frac{(A_i^\infty)^2 + (A_j^\infty)^2}{2} \right) \right) \\
& = \beta(A_i^\infty + A_j^\infty)(1-B_{ij}^\infty) - \beta^2 \frac{(A_i^\infty + A_j^\infty)^2}{2}(1-B_{ij}^\infty) \\
& \hspace{1cm} + \beta^2 \left\{ (\rho^\infty)^2 (1-B_{ij}^\infty) - \frac{(A_i^\infty - A_j^\infty)^2}{2} \right\} \\
&\geq \beta(A_i^\infty + A_j^\infty)(1-B_{ij}^\infty)\left(1-\beta \frac{A_i^\infty + A_j^\infty}{2} \right) \geq 0.
\end{aligned}
\end{align*}
This yields, 
\[ \mbox{either}~B_{ij}^\infty = 1 \quad \mbox{or} \quad \beta(A_i^\infty + A_j^\infty) = 2. \]
Since $0 < \beta \leq 1$ and $0 < A_i \leq 1$, the latter case implies 
\[ \beta = A_i^\infty = A_j^\infty = 1. \]
This implies $B_{ij}^\infty = 1$. Therefore,  for all cases, we get the desired estimate:
\[ B_{ij}^\infty = 1, \quad  i, j \in [N]. \] 
\end{proof}

\section{Discrete Lohe matrix flow on the unitary group} \label{sec:4}
\setcounter{equation}{0}
In this section, we present several discretization algorithms for the Lohe matrix model \eqref{A-2} using an exponential map such as the Lie group integrator method and splitting methods, and then we apply the aforementioned discretization algorithms to the Lohe matrix model to derive three different discrete models for the Lohe matrix model. More over, for the three discrete models, we study their emergent properties. 
\subsection{Discretization methods} \label{sec:4.1}
In this subsection, we study two discretization methods, namely ``{\it the Lie group integrator method}" and ``{\it operator splitting method}". The Euler method discussed in the previous section may not preserve the structure of a state and thus a projection is required. In what follows, we present discretization algorithms without a projection step.
\subsubsection{The Lie group integrator method} \label{sec:4.1.1}
Let $({\mathcal M}, \langle \cdot, \cdot \rangle_{{\mathcal M}})$ be a Riemannian manifold with metric $\langle \cdot, \cdot \rangle_{\mathcal M}$, and we consider the Cauchy problem:
\begin{equation} \label{D-0}
\begin{cases}
\displaystyle \dot{\bx}_i =  {\boldmath f}_i({\mathcal X}),\quad t > 0,~~i \in [N], \\
\displaystyle \bx_i(0) = \bx_i^0\in\mcal{M},
\end{cases}
\end{equation}
where we assumed
\begin{equation} \label{D-0-0}
{\boldmath f}_i({\mathcal X}) \in T_{\bx} \mcal{M} \quad \mbox{at any point $\bx \in M$},~~i \in [N].
\end{equation} 
As noticed before, the forward Euler scheme does not guarantee ${\mathcal X}(n+1) \in \mcal{M}^N$ even if ${\mathcal X}(n) \in {\mathcal M}^N$ in general.  Thus, we take the update of $\bx_i(n)$ to the tangential direction ${\boldmath f}_i({\mathcal X}(n))$ via an {\it exponential map} \cite{Bie, Do}.  Before we move on further, we briefly recall the exponential map below. We set the tangent bundle $T {\mathcal M}$ as 
\[ T{\mathcal M} := \{(\bx, \bv) \in {\mathcal M} \times T_{\bx} {\mathcal M}:~ \bx \in {\mathcal M},~~ \bv \in T_{\bx} {\mathcal M} \}, \]
where $T_{\bx} {\mathcal M}$ is the tangent space of $\mcal{M}$ at $\bx \in {\mathcal M}$.
Since we are interested in asymptotic behaviors, we only consider that the Riemannian manifold ${\mathcal M}$ is {\it geodesically complete} in the sense that for any two points $\bx, \by \in {\mathcal M}$, there exists a unique length minimizing geodesic between $\bx$ and $\by$. For a given $(\bx, \bv) \in {\mathcal M} \times T_{\bx}{\mathcal M}$, let $\gamma: [0, 1] \to {\mathcal M}$  be a locally length minimizing geodesic such that 
\begin{equation} \label{D-0-1}
\gamma(0) = \bx, \quad {\dot \gamma}(0) = \bv. 
\end{equation}
For $\bx \in {\mathcal M}$, the exponential map $\mbox{exp}_{\bx}:~T_{\bx} {\mathcal M} \to {\mathcal M}$ is defined as 
\[ \exp_{\bx}(\bv) :=  \gamma(1). \]
Note that for the trivial manifold ${\mathcal M} = {\mathbb R}^d$, one has 
\[ T_x {\mathcal M} \simeq {\mathbb R}^d, \quad  T{\mathcal M}   \simeq  {\mathbb R}^d \times  {\mathbb R}^d,  \]
and the geodesic $\gamma$ satisfying \eqref{D-0-1} and the exponential map are just a straight line passing through $\bx$ in the direction of $\bv$ and the sum of $\bx$ and $\bv$:
\[ \gamma(t) = \bx + t \bv, \quad t \in {\mathbb R}, \quad  \exp_{\bx} \bv = \gamma(1) = \bx + \bv. \]
Now, we propose the discrete model for \eqref{D-0} via the exponential map update  ${\mathcal X}(n)$ according to the following discrete map:
\begin{equation}
\begin{cases} \label{D-0-2}
\bx_i(n+1) = \exp_{\bx_i(n)}\Big (h {\boldmath f}_i({\mathcal X}(n)) \Big), \quad n \geq 0, \\
\bx_i(0) = \bx_i^0 \in\mcal{M}, \quad i \in [N].
\end{cases}
\end{equation}
By \eqref{D-0-0}, as long as system \eqref{D-0-2} admits a solution, one has
\[ {\mathcal X}(n)\in\mcal{M}^N, \quad n\geq 0.\]
Note that for ${\mathcal M} = {\mathbb R}^d$, discrete system \eqref{D-0-2} can cast as the first-order forward Euler discretization of \eqref{D-0}:
\[  {\bx}_i (n+1) = \bx_i(n) + h {\boldmath f}_i({\mathcal X}(n)), \quad n \geq 0,~~ i \in [N].  \]

\vspace{0.2cm}

Next, we return to our setting. Let $G$ be a Lie group, and consider the following Cauchy problem:
\begin{equation} \label{D-1}
\begin{cases}
\displaystyle \dot{\bx}_i = A_i({\mathcal X}) \bx_i,\quad t > 0, \\
\displaystyle \bx_i(0) = \bx_i^0 \in G, \quad i \in [N],
\end{cases}
\end{equation}
where we assumed
\begin{equation} \label{D-1-0}
 A_i({\mathcal X})\in \mathfrak{g} := T_{\text{id}} G \quad \mbox{at any point $\bx \in G$}. 
\end{equation}
Note that a Lie group equipped with a left-invariant metric is geodesically complete. For a Lie group case, discrete scheme \eqref{D-0-2} can be rewritten as
\begin{equation}
\begin{cases} \label{D-1-2}
\bx_i(n+1) = \exp\big (h A_i({\mathcal X}(n)) \big) \bx_i(n), \quad n \geq 0, \\
\bx_i(0) = \bx_i^0 \in G,~~ i \in [N].
\end{cases}
\end{equation}
On a matrix Lie group, the exponential map is equivalent to the matrix exponential. This numerical scheme for system \eqref{D-1} is often called the discrete model with the {\it ``Lie group integrator method"} \cite{C-M-O, Ib, I-M-N}.

\subsubsection{Operator splitting method} \label{sec:4.1.2} 
In this part, we consider the case in which the coefficient in front of $\bx_i$ in \eqref{D-1} can split into two parts, say
\begin{equation} \label{D-1-3}
\begin{cases}
\displaystyle \dot{\bx}_i = (A_i({\mathcal X}) + B_i({\mathcal X})) \bx_i,\quad t > 0,\\
\displaystyle \bx_i(0) = \bx_i^0 \in G,~~i \in [N],
\end{cases}
\end{equation}
where we also assume \eqref{D-1-0} for $A_i$ and $B_i$ for all $i \in [N]$. For constant matrices $A_i$ and $B_i$, the unique solution to \eqref{D-1-3} is given by
\[ \bx_i(t) = \exp( tA_i+ tB_i ) \bx_i^0, \quad t \geq 0,\quad i \in [N]. \]
If $A_i$ and $B_i$ are not commutative, it is very difficult to calculate the matrix exponential $ e^{tA_i+ tB_i}$. Thus, we propose another discrete model motivated by the Lie-Trotter product formula \cite{Tro}:
\begin{equation*}
e^{A_i+B_i} = \lim_{n\to\infty} \left(e^{A_i/n} e^{B_i/n}\right)^n.
\end{equation*}
Now, if we apply the Lie-trotter splitting method to \eqref{D-1-3}, one has 
\begin{equation} \label{D-1-4}
\begin{cases}
\displaystyle \bx_i(n+1) = \exp(h A_i({\mathcal X}(n)) )\exp(h B_i({\mathcal X}(n))) \bx_i(n) ,\quad n \geq 0, \\
\bx_i(0) = \bx_i^0 \in G,~~ i \in [N].
\end{cases}
\end{equation}
It is known that the Lie-trotter splitting scheme admits a first order local truncation error:
\[e^{h(A+B)} = e^{hA} e^{hB} + \mcal{O}(h).\]
Our discrete models are the first order approximation of the continuous model, hence the discretization admits second order local truncation error. The first order error from the operator splitting may be relatively large to the local truncation error from the time discretization. To get rid of such issues, we introduce the Strang splitting:
\[e^{h(A+B)} = e^{hA/2} e^{hB} e^{hA/2} + \mcal{O}(h^2). \]
See \cite{J-L, Mu, Str} for the detailed description of the Strang splitting and higher order splitting methods. The Strang splitting method applied to \eqref{D-1-3} yields
\begin{equation}
\begin{cases} \label{D-1-5}
\displaystyle \bx_i(n+1) = \exp \left(\frac{h}{2} A_i({\mathcal X}(n)) \right) \exp \Big(h B_i({\mathcal X}(n)) \Big) \exp \left(\frac{h}{2} A_i({\mathcal X}(n)) \right) \bx_i(n),~~ n \geq 0,\\
\displaystyle \bx_i(0) = \bx_i^0 \in G, \quad i \in [N].
\end{cases}
\end{equation}

\subsection{Discrete Lohe matrix models}  \label{sec:4.2}
In this subsection, we present three explicit discrete models for \eqref{A-2} using the discretization schemes discussed in the previous subsection. First, we rewrite the system \eqref{A-2} as in the form of  \eqref{D-1}:
\begin{equation} 
\begin{cases} \label{D-2}
\displaystyle \dot{U}_i = A_i(\mathcal{U}) U_i, \quad i \in [N], \\
\displaystyle A_i(\mathcal{U}) = -\mathrm{i}H_i + \frac{\kappa}{2}\big(U_c U_i^\dag - U_i U_c^\dag \big),
\end{cases}
\end{equation}
where $U_c := \frac{1}{N} \sum_{j=1}^{N} U_j$.

\vspace{0.2cm}

In what follows, we introduce three discretize schemes using Lie group integrator formula \eqref{D-1-2}, Lie-Trotter splitting formula \eqref{D-1-4}, and Strang splitting formula \eqref{D-1-5}. Without these formulas, we might use the projection algorithm(which applied in the Lohe sphere model) to discretize the Lohe matrix model. However, the projection of a general matrix onto the unitary group $\mathbb{U}(d)$ is incomparably harder, since we should use the SVD(singular value decomposition) to find the projection. However, the SVD has no explicit formula. This is the priimary reason why we can not find estimates for the discretized system via projection.

\vspace{0.2cm}

For the discretized schemes \eqref{D-1-2}, \eqref{D-1-4}, and \eqref{D-1-5}, we can make sufficient estimates for asymptotic behaviors. In particular,  for a homogeneous model, we obtain complete state aggregation as expected, whereas we obtain the orbital stability for a heterogeneous model. It is the main benefit of using the Lie group exponential maps.

\vspace{0.2cm}

\subsubsection{The discrete Lohe matrix model A}~If we apply the Lie group integrator method \eqref{D-1-2} to the Lohe matrix model \eqref{D-2}, one has the discrete Lohe matrix model A: 
\begin{align}
\begin{aligned} \label{D-2-2}
U_i (n+1) &= \exp \Big(h A_i (\mathcal{U}(n)) \Big) U_i(n) \\
 &=  {\exp\left(-\mathrm{i} H_i h + \frac{\kappa h}{2}(U_c(n) U_i^\dag (n) - U_i(n) U_c^\dag (n))\right) U_i(n)},  \quad i \in [N], ~~ n \geq 0.
\end{aligned}
\end{align}
Discrete model \eqref{D-2-2} can be reduced to the discrete Kuramoto model \eqref{C-5-1} in a special setting:
\begin{equation} \label{D-2-3}
U_i = e^{-\mathrm{i}\theta_i},\quad H_i = \nu_i, \quad i \in [N].
\end{equation}
It follows from \eqref{D-2-2} and \eqref{D-2-3} that 
\begin{align}
\begin{aligned} \label{D-2-4}
& \exp(-\mathrm{i}\theta_i(n+1)) = U_i(n+1) = \exp \bigg(-\mathrm{i}H_i h + \frac{\kappa h}{2N} \sum_{j=1}^N (U_j U_i^\dag - U_i U_j^\dag) \bigg) U_i(n) \\
& \hspace{.7cm} = \exp \bigg( -\mathrm{i}\nu_i h + \frac{\kappa h}{2N} \sum_{j=1}^N (e^{-\mathrm{i}(\theta_j(n) - \theta_i(n))} - e^{-\mathrm{i}(\theta_i(n) - \theta_j(n))})\bigg) \exp(-\mathrm{i}\theta_i(n)) \\
& \hspace{.7cm} = \exp\bigg(-\mathrm{i}\Big[ \theta_i(n) + \nu_i h + \frac{\kappa h}{N} \sum\limits_{j=1}^N \sin(\theta_j(n) - \theta_i(n)) \Big] \bigg).
\end{aligned}
\end{align}
By comparing the exponents on both sides of \eqref{D-2-4}, one obtains the discrete Kuramoto model \eqref{C-5-1}. \newline

\vspace{0.2cm}

\subsubsection{The discrete Lohe matrix model B}~If we apply the Lie-Trotter splitting scheme \eqref{D-1-4} for system \eqref{D-2},  one obtains the second discrete model: 
\begin{equation} \label{D-2-5}
\displaystyle U_i (n+1)=  \exp(-\mathrm{i} H_i h)  \exp \left(\frac{\kappa h}{2}(U_c(n) U_i^\dag (n) - U_i(n) U_c^\dag (n))\right) U_i(n),
\end{equation}
for $i \in [N]$ and $n \geq 0$. Note that this model \eqref{D-2-5} can not be followed from \eqref{D-2-2}, unless $-\mathrm{i} H_i h$ and $\frac{\kappa h}{2}(U_c(n) U_i^\dag (n) - U_i(n) U_c^\dag (n))$ commute in general. 

\vspace{0.2cm}

\subsubsection{The discrete Lohe matrix model C}~Again, we apply the Strang splitting scheme \eqref{D-1-5} for system \eqref{D-2} to get the third discrete model:
\begin{equation} \label{D-2-6}
\displaystyle U_i (n+1)=  \exp \left(-\frac{\mathrm{i} H_i h}{2} \right)  \exp \left(\frac{\kappa h}{2}(U_c(n) U_i^\dag (n) - U_i(n) U_c^\dag (n))\right) \exp \left(-\frac{\mathrm{i} H_i h}{2} \right) U_i(n),
\end{equation}
for $ i \in [N]$ and $n \geq 0$.

\vspace{0.2cm}

Note that for zero free flows:
\[ H_i = O,  \quad i \in [N]. \]
all discrete Lohe matrix models reduced to same form \eqref{D-3}. In the following three sections, we study emergent dynamics of the discrete Lohe matrix models one by one. 

\section{A homogeneous matrix ensemble}
}  \label{sec:5}
\setcounter{equation}{0} 
In this section, we study emergent behaviors of a homogeneous ensemble with the same hamiltonian $H$. Thanks to the solution splitting property, at the level of continuous system, we can assume the common hamiltonian to be zero. In this case, all three discrete Lohe matrix models coincide and it reads as 
\begin{equation} \label{D-3}
\displaystyle U_i (n+1)  = \exp\left( \frac{\kappa h}{2}(U_c(n) U_i^\dag (n) - U_i(n) U_c^\dag (n))\right) U_i(n),\quad i \in [N] ,~~ n \geq 0.
\end{equation}

\vspace{0.2cm}

\subsection{Preliminary lemmas} \label{sec:5.1}
In this subsection, we study several lemmas to be crucially used in the next subsection. First, we study the relations between the Frobenius norm and operator norm for complex valued square matrix. 
\begin{lemma} \label{L5.1} 
Let $A$ and $B$ be square matrices in ${\mathbb C}^{d \times d}$ and $U \in \mbf{U}(d)$. Then, one has 
\[ \norm{AB}_\F \leq \norm{A}_{\op} \norm{B}_\F,\quad \norm{AB}_\F \leq \norm{A}_\F \norm{B}_{\op}, \quad \norm{AU}_\F = \norm{UA}_\F = \norm{A}_\F. \]
\end{lemma}
\begin{proof}
Let $B = \left[\begin{array}{c|c|c} b_1 & \cdots & b_d \end{array}\right]$ for column vectors $b_1,\cdots,b_d$. Then, we have
\begin{align*}
& \norm{AB}_\F^2 = \sum_{i=1}^d \norm{Ab_i}^2 \leq \norm{A}_{\op}^2 \sum_{i=1}^d \norm{b_i}^2 = \norm{A}_{\op}^2 \norm{B}_\F^2, \\
& \norm{AB}_\F =\snorm{B^\dag A^\dag}_\F \leq \snorm{B^\dag}_{\op} \snorm{A^\dag}_\F = \snorm{B}_{\op} \snorm{A}_\F, \\
& \norm{AU}_\F^2 = \tr{AUU^\dag A^\dag} = \tr{AA^\dag} = \norm{A}_\F^2 . \qedhere
\end{align*}
\end{proof}
\begin{lemma} \label{L5.2}
Let $A_1,\cdots,A_k$ and $B_1,\cdots, B_k$ be square matrices in ${\mathbb C}^{d\times d}$. Then, one has the following assertions:
\begin{align*}
& \|A_1A_2\cdots A_k\|_\F \leq \|A_1\|_\F \|A_2\|_\F \cdots \|A_k\|_\F,\\
& |\mathrm{tr}(A_1A_2\cdots A_k)|\leq \norm{A_1}_\F \norm{A_2}_\F \cdots  \norm{A_k}_\F.
\end{align*}
\end{lemma}
\begin{proof} It suffices to show that the both assertions hold for two matrices:
\[ \norm{AB}_\F \leq \norm{A}_\F \cdot \norm{B}_\F, \quad |\mathrm{tr}(AB)| \leq \norm{A}_\F \cdot \norm{B}_\F.  \]
The general case can be treated using mathematical induction. We set 
\[ A = (a_{ij}), \quad B = (b_{ij}), \]
where $[A]_{ij}$ denotes the $(i,j)$-component of the matrix A. \newline

\noindent (i)~We use the Cauchy-Schwarz inequality to get 
\[  |[AB]_{ij}| = \Big| \sum_{k=1}^{N} a_{ik} b_{kj} \Big| \leq \sum_{k=1}^{N} |a_{ik} | \cdot |b_{kj}| \leq \Big( \sum_{k=1}^{N} |a_{ik}|^2  \Big)^{\frac{1}{2}} \cdot \Big( \sum_{k=1}^{N} |b_{kj}|^2  \Big)^{\frac{1}{2}}.   \]
This yields
\begin{align*}
\begin{aligned}
\norm{AB}_\F^2 &= \sum_{i,j=1}^{N}  \Big |[AB]_{ij} \Big|^2 \leq  \sum_{i,j=1}^{N} \Big( \sum_{k=1}^{N} |a_{ik}|^2  \Big) \cdot  \Big( \sum_{k=1}^{N} |b_{ik}|^2  \Big) \\
&\leq \Big( \sum_{i,k=1}^{N} |a_{ik}|^2  \Big) \cdot  \Big( \sum_{j,k=1}^{N} |b_{ik}|^2  \Big) = \norm{A}^2_\F \cdot \norm{B}^2_\F.
\end{aligned}
\end{align*}

\noindent (ii)~From simple calculations, we get
\begin{align*}
 |\tr{AB} | &= \Big| \sum_{i=1}^{d} [AB]_{ii} \Big| \leq \sum_{i, k = 1}^d  |[A]_{ik}| \cdot |[B]_{ki}|  \\
 & \leq \Big( \sum_{i, k = 1}^d  |[A]_{ik}|^2  \Big)^{\frac{1}{2}} \cdot \Big( \sum_{i, k = 1}^d  |[B]_{ki}|^2  \Big)^{\frac{1}{2}} \leq \norm{A}_\F \cdot \norm{B}_\F. \qedhere
\end{align*}
\end{proof}
\begin{remark}
For $k=1$, Lemma \ref{L5.2} can be reduced as follows:
\[ |\mathrm{tr}(A) |= \Big| \sum_{i=1}^{d} a_{ii} \Big| \leq \sum_{i=1}^{d} |a_{ii}| \leq \sqrt{d} \Big( \sum_{i=1}^{d} |a_{ii}|^2 \Big)^{\frac{1}{2}} \leq  \sqrt{d}  \norm{A}_\F. \]
\end{remark}

\begin{lemma} \label{L5.3}
Let $A_1,\cdots,A_k$ and $B_1,\cdots,B_k$ be square matrices in ${\mathbb C}^{d\times d}$ such that 
\[ \norm{A_i}_\op \leq 1, \quad \norm{B_i}_\op \leq 1, \quad i \in [K]. \]
Then, one has 
\[ \norm{A_1\cdots A_k - B_1\cdots B_k}_\F \leq \sum_{l=1}^k \norm{A_l - B_l}_\F. 
\]
\end{lemma}
\begin{proof} As in Lemma \ref{L5.2}, it is sufficient to check that the assertion holds for $k = 2$. The general case can be made using the mathematical induction and the result of Lemma \ref{L5.2}. Suppose that $A_1, A_2, B_1$ and $B_2$ are square matrices in $\C^{d\times d}$ such that 
\[ \norm{A_i}_{\op} \leq 1, \quad  \norm{A_i}_{\op} \leq 1,\quad i = 1,2. \]
Then, we use Lemma \ref{L5.1} and Lemma \ref{L5.2} to get 
\begin{align*}
\norm{A_1 A_2 - B_1 B_2}_\F &= \norm{(A_1 - B_1) A_2 + B_1 (A_2 - B_2)}_\F \\
&\leq \norm{(A_1 - B_1) A_2}_\F + \norm{B_1 (A_2 - B_2)}_\F \\
& \leq \norm{A_1 - B_1}_\F \cdot \norm{A_2}_\op + \norm{B_1}_\F \cdot \norm{A_2 - B_2}_\F \\
& \leq \norm{A_1 - B_1}_\F + \norm{A_2 - B_2}_\F. \qedhere
\end{align*}
\end{proof}

\vspace{0.2cm}

For $n\geq 0,~\mathcal{U}(n)=(U_1(n),\cdots,U_N(n)) \in (\mbf{U}(d))^N$ and $i \in [N]$, we set
\begin{equation} \label{D-3-0}
\mcal{D}(n) := \max\limits_{1 \leq i,j \leq N} \norm{U_i(n) - U_j(n)}_\F,  \quad \Delta_i := \frac{1}{2}(U_c U_i^\dag - U_i U_c^\dag).
\end{equation}
Note that system $\eqref{D-3}_1$ can be rewritten as: 
\begin{equation*} \label{D-3-1}
U_i(n+1) = \exp\left(\frac{\kappa h}{2}(U_c(n) U_i^\dag (n) - U_i(n) U_c^\dag (n))\right) U_i(n) = e^{\beta \Delta_i(n)} U_i(n).
\end{equation*}
\begin{lemma} \label{L5.4} 
Let $U_1,\cdots,U_N$ be $d\times d$ unitary matrices. Then, one has the following estimates:~for $i, j \in [N]$,
\begin{align*}
\begin{aligned}
& (i)~\snorm{U_i U_j^\dag - U_j U_i^\dag}_\F\leq 2 \norm{U_i - U_j}_\F, \quad \norm{\Delta_i}_\F \leq \frac{1}{N} \sum\limits_{k=1}^N \norm{U_k -U_i}_\F, \\
& (ii)~\norm{\Delta_i}_{\op} \leq \norm{U_c}_{\op} \leq 1, ~~ \norm{\Delta_i - \Delta_j}_\F  \leq \norm{U_i - U_j}_\F, ~~ \norm{\Delta_i - \Delta_j}_{\op} \leq \norm{U_i - U_j}_{\op}.
\end{aligned}
\end{align*}
\end{lemma}
\begin{proof}
\noindent (i)~Note that 
\[ U_i U_j^\dag - U_j U_i^\dag = (U_i - U_j) U_i^\dag - U_i(U_i^\dag - U_j^\dag). \]
This relation and Lemma \ref{L5.1} imply 
\[ \norm{U_i U_j^\dag - U_j U_i^\dag}_\F \leq \snorm{(U_i - U_j) U_i^\dag }_\F + \snorm{U_i(U_i^\dag - U_j^\dag)}_\F = 2\norm{U_i - U_j}_\F. \] 
On the other hand, we use the above result and defining relation $\Delta_i$ in \eqref{D-3-0}  to get 
\[\snorm{\Delta_i}_\F \leq \frac{1}{2N} \sum_{k=1}^N \norm{U_k U_i^\dag - U_i U_k^\dag}_\F \leq \frac{1}{N} \sum_{k=1}^N \norm{U_k - U_i}_\F\]
Moreover, one has 
\begin{equation} \label{D-3-2}
\norm{\Delta_i}_{\op} \leq \frac{1}{2} \big(\norm{U_c}_{\op} \snorm{U_i^\dag}_{\op} + \norm{U_i}_{\op} \snorm{U_c^\dag}_{\op} \big) = \norm{U_c}_{\op} \leq 1.
\end{equation}

\vspace{0.1cm}

\noindent(iii) Note that
\[
\Delta_i - \Delta_j = \frac{1}{2} \big( U_c (U_i - U_j)^\dag - (U_i - U_j) U_c^\dag \big).
\]
Then, we use Lemma \ref{L5.1} and \eqref{D-3-2} to find
\[ \norm{\Delta_i - \Delta_j}_\F \leq \norm{U_c}_{\op} \cdot\norm{U_i - U_j}_\F \leq \norm{U_i - U_j}_\F. \] 
Similarly, we obtain 
\[ \norm{\Delta_i - \Delta_j}_{\op} \leq \norm{U_c}_{\op} \cdot\norm{U_i - U_j}_\F \leq \norm{U_i - U_j}_{\op}. \qedhere \]
\end{proof}

\subsection{Complete state aggregation} \label{sec:5.2}
In this subsection, we derive complete state aggregation via an exponential decay of ${\mathcal D}(n)$ for \eqref{D-3}. \newline

Consider the variation of the ensemble diameter $\mcal{D}$, i.e., $\mcal{D}(n+1) - \mcal{D}(n)$. For this, we consider the estimate for 
\[ \norm{U_i(n+1) - U_j(n+1)}_\F^2 - \norm{U_i(n) - U_j(n)}_\F^2. \]

For notational simplicity, until end of this section, unless stated otherwise, the states $U_i$, and the functionals $U_c$, $\Delta_i$ are evaluated at the $n$-th time step.
\begin{align*}
\begin{aligned}
&\norm{U_i(n+1) - U_j(n+1)}_\F^2 - \norm{U_i(n) - U_j(n)}_\F^2 \\
& \hspace{0.5cm} =\; \tr{-U_i(n+1) U_j^\dag(n+1) - U_j(n+1) U_i^\dag(n+1) + U_i(n) U_j^\dag (n) + U_j(n) U_i^\dag (n)} \\
& \hspace{0.5cm} = \; \tr{-e^{\beta\Delta_i} U_i U_j^\dag e^{-\beta\Delta_j} - e^{\beta\Delta_j} U_j U_i^\dag e^{-\beta\Delta_i} + U_i U_j^\dag + U_j U_i^\dag } \quad \text{ at }n \text{-th time step}\\
& \hspace{0.5cm} =\; \tr{- e^{-\beta\Delta_j} e^{\beta\Delta_i} U_i U_j^\dag - e^{-\beta\Delta_i} e^{\beta\Delta_j} U_j U_i^\dag + U_i U_j^\dag + U_j U_i^\dag } \\
& \hspace{0.5cm} =\; -\tr{(e^{-\beta\Delta_j} e^{\beta\Delta_i} - I) U_i U_j^\dag + (e^{-\beta\Delta_i} e^{\beta\Delta_j} -I) U_j U_i^\dag} \\
& \hspace{0.5cm} =\; - \beta \underbrace{\tr{(\Delta_i - \Delta_j) U_i U_j^\dag + (\Delta_j - \Delta_i) U_j U_i^\dag}}_{=:\mcal{A}_1} \\
& \hspace{0.7cm} - \frac{\beta^2}{2!} \underbrace{\tr{(\Delta_i^2 - 2\Delta_j \Delta_i + \Delta_j^2) U_i U_j^\dag + ( \Delta_j^2 - 2\Delta_i \Delta_j + \Delta_i^2) U_j U_i^\dag}}_{=:\mcal{A}_2} \\
& \hspace{0.7cm} - \frac{\beta^3}{3!} \underbrace{\tr{(\Delta_i^3 - 3\Delta_j \Delta_i^2 + 3\Delta_j^2 \Delta_i - \Delta_j^3) U_i U_j^\dag + (\Delta_j^3 - 3\Delta_i \Delta_j^2 + 3\Delta_i^2 \Delta_j - \Delta_i^3) U_j U_i^\dag}}_{=:\mcal{A}_3} \\
& \hspace{0.7cm}- \cdots,
\end{aligned}
\end{align*}
where in the last equality, we used the expansion of the matrix exponential to get 
\begin{align}
\begin{aligned} \label{D-8} 
    e^{-\beta \Delta_j} e^{\beta \Delta_i} &= \left(\sum_{k=0}^\infty \frac{1}{k!} (-\beta\Delta_j)^k \right) \left( \sum_{l=0}^\infty \frac{1}{l!} (\beta\Delta_i)^l \right) \\
    &= \sum_{k,l=0}^\infty \frac{1}{k!l!} (-1)^k \beta^{k+l} \Delta_j^k \Delta_i^l  = \sum_{m=0}^\infty \frac{\beta^m}{m!} \left( \sum_{k=0}^m \frac{m!}{k!(m-k)!} (-1)^k \Delta_j^k \Delta_i^{m-k} \right).
\end{aligned}
 \end{align}
Therefore, we get
\begin{equation} \label{D-9}
\norm{U_i(n+1) - U_j(n+1)}_\F^2 = \norm{U_i(n) - U_j(n)}_\F^2 -\sum_{m=1}^\infty \frac{\beta^m}{m!} \mcal{A}_m,
\end{equation}
where
\begin{align*}
\mcal{A}_m &= \underbrace{\text{tr}\Bigg\{ \bigg( \sum_{k=0}^m \frac{m!}{k!(m-k)!} (-1)^k \Delta_j^k \Delta_i^{m-k} \bigg) U_i U_j^\dag \Bigg\}}_{\mcal{A}_m^+} + \underbrace{\text{tr} \Bigg\{ \bigg( \sum_{k=0}^m \frac{m!}{k!(m-k)!} (-1)^k \Delta_i^k \Delta_j^{m-k} \bigg) U_j U_i^\dag\Bigg\}}_{\mcal{A}_m^-}.
\end{align*}
To majorize the R.H.S. of \eqref{D-9} by  term involving with $\norm{U_i(n) - U_j(n)}_\F^2$, we estimate $\mathcal{A}_1$ and $\mathcal{A}_m$ with $m\geq2$ separately in the following two lemmas.
\begin{lemma}\label{L5.5}
The term $\mcal{A}_1$ satisfies the following estimates:
\begin{align*}
\begin{aligned}
& (i)~\mcal{A}_1 \geq 2\norm{U_i - U_j}_\F^2 - \sqrt{d} \big(\max_{1 \leq k,l \leq N}\norm{U_k - U_l}_{\op}^2 \big) \norm{U_i - U_j}_\F^2. \\
& (ii)~\mcal{A}_1 \geq 2\norm{U_i - U_j}_\F^2 - \big(\max_{1 \leq k,l \leq N}\norm{U_k - U_l}_\F^2 \big) \norm{U_i - U_j}_\F^2.
\end{aligned}
\end{align*}
\end{lemma}
\begin{proof}
First, we rewrite ${\mathcal A}_1$ as follows:
{\allowdisplaybreaks
 \begin{align}
 \begin{aligned} \label{D-9-1}
\mcal{A}_1 &= \tr{(\Delta_i - \Delta_j) U_i U_j^\dag + (\Delta_j - \Delta_i) U_j U_i^\dag} = \tr{(\Delta_i - \Delta_j)(U_i U_j^\dag - U_j U_i^\dag)}\\
&= \frac{1}{2N} \sum_{k=1}^N\tr{\big(U_k (U_i^\dag - U_j^\dag) - (U_i - U_j) U_k^\dag \big)(U_i U_j^\dag - U_j U_i^\dag)} \\
&= \frac{1}{2N} \sum_{k=1}^N \left[ \tr{U_k (U_i^\dag - U_j^\dag) (U_i U_j^\dag - U_j U_i^\dag)} - \tr{ (U_i U_j^\dag - U_j U_i^\dag)(U_i - U_j) U_k^\dag} \right]\\
&= \frac{1}{2N} \sum_{k=1}^N \tr{U_k (U_i^\dag + U_j^\dag)(U_i - U_j)(U_i^\dag - U_j^\dag)} + \tr{(U_i - U_j)(U_i^\dag - U_j^\dag)(U_i + U_j)U_k^\dag} \\
&= \frac{1}{2N} \sum_{k=1}^N \tr{(U_k U_i^\dag + U_k U_j^\dag + U_i U_k^\dag + U_j U_k^\dag)(U_i - U_j)(U_i^\dag - U_j^\dag)} \\\
&= \frac{1}{2N} \sum_{k=1}^N \left[ 4\; \tr{(U_i - U_j)(U_i^\dag - U_j^\dag)} \right. \\
& \hspace{0.5cm} \left. - \tr{ \big\{ (U_k - U_i)(U_k^\dag - U_i^\dag) + (U_k - U_j)(U_k^\dag - U_j^\dag) \big\} (U_i - U_j)(U_i^\dag - U_j^\dag) } \right] \\
&= 2 \norm{U_i - U_j}_\F^2 - \frac{1}{2N} \sum_{k=1}^N \left[ \tr{(U_k - U_i)(U_k^\dag - U_i^\dag) (U_i - U_j)(U_i^\dag - U_j^\dag)} \right. \\
& \qquad\qquad\qquad\qquad\qquad\quad \left. + \tr{(U_k - U_j)(U_k^\dag - U_j^\dag) (U_i - U_j)(U_i^\dag - U_j^\dag)} \right],
\end{aligned}
\end{align}}
where in the fourth equality, we used the following identities for unitary matrices $U_i$ and $U_j$:
\begin{align*}
(U_i^\dag - U_j^\dag) (U_i U_j^\dag - U_j U_i^\dag) &= (U_i^\dag + U_j^\dag)(U_i - U_j)(U_i^\dag - U_j^\dag), \\
(U_i U_j^\dag - U_j U_i^\dag)(U_i - U_j) &= - (U_i - U_j)(U_i^\dag - U_j^\dag)(U_i + U_j).
\end{align*}
Now we use Lemma \ref{L5.1}  to see
\begin{align}
\begin{aligned} \label{D-9-2}
& \abs{\tr{(U_k - U_i)(U_k^\dag - U_i^\dag) (U_i - U_j)(U_i^\dag - U_j^\dag)}} \\
& \hspace{.5cm} \leq \norm{(U_k^\dag - U_i^\dag) (U_i - U_j)}_\F \norm{(U_i^\dag - U_j^\dag)(U_k - U_i)}_\F \\
& \hspace{.5cm} \leq \sqrt{d} \norm{U_k - U_i}_\op^2 \cdot \norm{U_i - U_j}_\F^2 \leq \sqrt{d}\, \Big( \max_{1 \leq i, k \leq N} \norm{U_k - U_i}_\op^2 \Big)  \cdot \norm{U_i - U_j}_\F^2.
\end{aligned}
\end{align}
Similarly, one has 
\begin{equation} \label{D-9-3}
\abs{\tr{(U_k - U_i)(U_k^\dag - U_i^\dag) (U_i - U_j)(U_i^\dag - U_j^\dag)}} \leq \Big( \max_{1 \leq i, k \leq N} \norm{U_k - U_i}_\F^2  \Big) \cdot \norm{U_i - U_j}_\F^2.
\end{equation}
Finally, we combine \eqref{D-9-1}, \eqref{D-9-2} and \eqref{D-9-3} to derive desired estimates. 
\end{proof}

\vspace{0.5cm}

Next, we estimate the term $\mcal{A}_m$ with $m\geq 2$ in the following lemma.
\begin{lemma} \label{L5.6} For $m\geq 2$, $\mcal{A}_m$ satisfies 
\[
\abs{\mcal{A}_m} \leq (m+1) 2^{m-1} \norm{U_i - U_j}_\F^2.
 \]
\end{lemma}
\begin{proof}

\noindent We first claim that $\mcal{A}_{m+1}$ can be written as follows:
\begin{align}
\begin{aligned} \label{D-10}
\mcal{A}_{m+1} &= \tr{(\Delta_i - \Delta_j) \bigg( \sum_{l=0}^m \frac{m!}{l!(m-l)!} (-1)^l (\Delta_i^{m-l} U_i U_j^\dag \Delta_j^l - \Delta_j^{m-l} U_j U_i^\dag \Delta_i^l) \bigg)} \\
&= \underbrace{\tr{(\Delta_i - \Delta_j) \sum_{l=0}^m \frac{m!}{l!(m-l)!} (-1)^l \Delta_i^{m-l} U_i U_j^\dag \Delta_j^l}}_{\mcal{B}_{m+1}^+}  \\
&\qquad + \underbrace{\tr{(\Delta_j - \Delta_i) \sum_{l=0}^m \frac{m!}{l!(m-l)!} (-1)^l \Delta_j^{m-l} U_j U_i^\dag \Delta_i^l}}_{\mcal{B}_{m+1}^-}.
\end{aligned} 
\end{align}
By the cyclicality of the trace, we have
\begin{align*}
\begin{aligned}
\mcal{B}_{m+1}^+ &= \tr{\sum_{l=0}^m \frac{m!}{l!(m-l)!} (-1)^l \Delta_j^l \Delta_i^{m+1-l} U_i U_j^\dag - \sum_{l=0}^m \frac{m!}{l!(m-l)!} (-1)^l \Delta_j^{l+1} \Delta_i^{m-l} U_i U_j^\dag} \\
&= \text{tr} \left( \sum_{l=0}^m \frac{m!}{l!(m-l)!} (-1)^l \Delta_j^l \Delta_i^{m+1-l} U_i U_j^\dag \right. - \left. \sum_{l=1}^{m+1} \frac{m!}{(l-1)!(m+1-l)!} (-1)^{l-1} \Delta_j^{l} \Delta_i^{m+1-l} U_i U_j^\dag \right) \\
&= \tr{\sum_{k=0}^{m+1} \bigg(\frac{m!}{k!(m-k)!} + \frac{m!}{(k-1)!(m+1-k)!} \bigg) (-1)^k \Delta_j^k \Delta_i^{m+1-k} U_i U_j^\dag }\\
&= \tr{\sum_{k=0}^{m+1} \frac{(m+1)!}{k!(m+1-k)!} (-1)^k \Delta_j^k \Delta_i^{m+1-k}U_i U_j^\dag} = \mcal{A}_{m+1}^+.
\end{aligned}
\end{align*}
Since we have $\mcal{B}_{m+1}^+ = \mcal{A}_{m+1}^+$, the formula $\mcal{B}_{m+1}^- = \mcal{A}_{m+1}^-$ immediately follows by interchanging the role of the indices $i$ and $j$. Therefore, we get \eqref{D-10} as a result.
\[\text{R.H.S. of \eqref{D-10}} = \mcal{B}_{m+1}^+ + \mcal{B}_{m+1}^- = \mcal{A}_{m+1}^+ + \mcal{A}_{m+1}^- = \mcal{A}_{m+1}. \]
By Lemma \ref{L5.1}, we have
\begin{align*}
\begin{aligned}
\norm{\Delta_i^m - \Delta_j^m}_\F &= \norm{\sum_{l=0}^{m-1} \Delta_i^l (\Delta_i - \Delta_j) \Delta_j^{m-1-l}}_\F  \\
& \leq \sum_{l=0}^{m-1} \norm{\Delta_i}_{\op}^l \norm{\Delta_i - \Delta_j}_\F \norm{\Delta_j}_{\op}^{m-1-l} \leq m \norm{U_i - U_j}_\F
\end{aligned}
\end{align*}
for all $m \geq 1$. Hence, by Lemma \ref{L5.3}, we have
\begin{align*}
\begin{aligned}
& \norm{\Delta_i^{m-l} U_i U_j^\dag \Delta_j^l - \Delta_j^{m-l} U_j U_i^\dag \Delta_i^l}_\F \\
& \hspace{0.2cm} \leq \; \norm{\Delta_i^{m-l} - \Delta_j^{m-l}}_\F + \norm{U_i U_j^\dag - U_j U_i^\dag}_\F + \norm{\Delta_i^{l} - \Delta_j^{l}}_\F \\
& \hspace{0.2cm} \leq \; (m-l) \norm{U_i - U_j}_\F + 2 \norm{U_i - U_j}_\F + l \norm{U_i - U_j}_\F = (m+2) \norm{U_i - U_j}_\F
\end{aligned}
\end{align*}
for all $0\leq l\leq m$. Therefore, we have
\begin{align*}
    \abs{\mcal{A}_{m+1}} &= \tr{(\Delta_i - \Delta_j) \left( \sum_{l=0}^m \frac{m!}{l!(m-l)!} \left(\Delta_i^{m-l} U_i U_j^\dag \Delta_j^l - \Delta_j^{m-l} U_j U_i^\dag \Delta_i^l\right) \right)} \\
    &\leq \norm{\Delta_i - \Delta_j}_\F \left(\sum_{l=0}^m \frac{m!}{l!(m-l)!} \norm{\Delta_i^{m-l} U_i U_j^\dag \Delta_j^l - \Delta_j^{m-l} U_j U_i^\dag \Delta_i^l}_\F \right) \\
    &\leq \norm{U_i - U_j}_\F \left(\sum_{l=0}^{m} \frac{m!}{l!(m-l)!} \right) \cdot (m+2) \norm{U_i - U_j}_\F \\
    & = (m+2)2^m \norm{U_i - U_j}_\F^2.
\end{align*}
This yields the desired result.
\end{proof}
Now, we are ready to provide complete aggregation for a homogeneous ensemble. Let  $\beta_0$ be the positive solution of the equation:
\[
e^{2\beta} + \frac{e^{2\beta}-1}{2\beta} = 4.
\]

\begin{theorem} \label{T5.1}
\emph{(Complete state aggregation)}
Suppose system parameters and initial data satisfy
\[ H_i = O, \quad i \in [N], \quad 0<\beta =\kappa h < \beta_0 \approx 0.437864, \quad \mcal{D}(0) < \sqrt{ 4 - e^{2\beta} - \frac{e^{2\beta}-1}{2\beta}}, \]
and let $\{U_i (n) \}$ be a solution to \eqref{D-3}. Then, complete state aggregation emerges asymptotically:
\[\lim_{n\to\infty} \max_{1\leq i, j \leq N} \norm{U_i(n) - U_j(n)}_\F = 0.\]
Moreover, the convergence rate is exponential.
\end{theorem}
\begin{proof} We use the mathematical induction. Suppose that $\norm{U_i(n) - U_j(n)}_\F$ is monotonically decreasing up to the $n$-th time step. Then, one has 
\[ \mcal{D}(n) \leq \mcal{D}(0). \]
Now we claim:
\[  \mcal{D}(n+1) \leq (1-C) \mcal{D}(n), \quad n\geq 0 \]
for a positive constant $C$ independent of $n$. \newline

By \eqref{D-9} and the estimates of $\mcal{A}_m$'s (Lemma \ref{L5.5} and Lemma \ref{L5.6}), we get 
\begin{align} 
\begin{aligned} \label{D-12}
& \lVert U_i(n+ 1) - U_j(n+1) \rVert_\F^2 - \lVert U_i(n) - U_j(n) \rVert_\F^2 = -\beta \mcal{A}_1 - \sum_{m=2}^\infty \frac{\beta^m}{m!} \mcal{A}_m \\
& \hspace{0.5cm} \leq - 2\beta \norm{U_i(n) - U_j(n)}_\F^2 + \beta\mcal{D}(n)^2 \norm{U_i(n) - U_j(n)}_\F^2 \\
& \hspace{0.9cm}  + \sum_{m=2}^\infty \frac{\beta^m}{m!} \cdot (m+1)2^{m-1} \norm{U_i(n) - U_j(n)}_\F^2\\
& \hspace{0.5cm}= - \beta \bigg( 2 - \mcal{D}(n)^2 - \sum_{m=2}^\infty \frac{(2\beta)^{m-1}}{m!}(m+1) \bigg) \norm{U_i(n)- U_j(n)}_\F^2 \\
& \hspace{0.5cm} \leq - \beta \bigg( 4 - e^{2\beta} - \frac{e^{2\beta}-1}{2\beta} - (\mcal{D}(0))^2 \bigg) \norm{U_i(n)- U_j(n)}_\F^2 \leq 0.
\end{aligned}
\end{align}
The above inequatliy implies that there exists a positive constant $C$ independent of $i,j,n$ such that 
\[\norm{U_i(n+1) - U_j(n+1)}_\F^2 \leq (1-C)^2 \norm{U_i(n) - U_j(n)}_\F^2 \]
for all $i,j \in [N]$ and $n\geq 0$.
\end{proof}
\begin{remark}
By the result of this theorem, all discrete Lohe matrix models exhibit exponential aggregation for a homogeneous ensemble. 
\end{remark}

\section{A hetrogeneous matrix ensemble} \label{sec:6}
\setcounter{equation}{0}
In this section, we study emergent behaviors of two discrete Lohe matrix models \eqref{D-2-5} and \eqref{D-2-6} with heterogeneous ensemble:
\[ {\mathcal D}({\mathcal H}) > 0. \]
Unfortunately, we do not have emergent dynamics estimate for the discrete Lohe matrix model A yet. Hence in the following two subsections, we consider the discrete Lohe matrix models B and C, separately. For these two discrete models, we study the following three estimates: Let $\{{\mathcal U}(n) \}$ be a discrete Lohe matrix flow.
\begin{itemize}
\item
Estiamte 1 (existence of positively invariant set):~ there exist positive numbers $n_0$ and $\alpha$ such that 
\[ U_i(n) \in {\mathcal B}(\alpha), \quad \forall~n \geq n_0, \quad i \in [N]. \]

\vspace{0.1cm}

\item
Estimate 2 (orbital stability):~for another discrete Lohe matrix flow $\{ {\tilde {\mathcal U}}(n) \}$, their configuration shapes are asymptotically the same in the sense that 
\[ \lim_{n \to \infty}  \max_{1 \leq i,j \leq N} \snorm{U_i(n) U_j^\dag(n) - \tilde{U}_i(n) \tilde{U}_j^\dag(n)}_\F = 0. \]

\vspace{0.1cm}

\item
Estimate 3 (existence of state-locking state):~for $i, j \in [N]$, the quadratic state $U_i U_k^{\dag}$ converges as $n \to \infty$.
\end{itemize}
In the following two lengthy subsections, we study the above three estimates for each discrete Lohe matrix model. \newline
\subsection{The discrete Lohe matrix model B} \label{sec:6.1}
Consider the Cauchy problem to the discrete Lohe matrix model B:
\begin{equation}
\begin{cases} \label{D-12-0}
\displaystyle U_i (n+1)=  \exp(-\mathrm{i} H_i h)  \exp \left(\frac{\kappa h}{2}(U_c(n) U_i^\dag (n) - U_i(n) U_c^\dag (n))\right) U_i(n), \\
U_i(0) = U^0_i.
\end{cases}
\end{equation}
\subsubsection{Positively invariant set} \label{sec:6.1.1} In this part, we study the existence of positively invariant set.  For this, we set
\begin{equation*}
{\mathcal D}({\mathcal H}) := \max\limits_{1\leq i,j \leq N} \norm{H_i - H_j}_\F\quad \text{and}\quad\Lambda(\beta) :=\begin{cases}
 4-e^{2\beta} - \frac{e^{2\beta}-1}{2\beta}\quad&\text{if }\beta\neq0,\\
 2&\text{if }\beta=0.
 \end{cases}
\end{equation*}
Then $\Lambda(\beta)$ is a decreasing function of $\beta>0$ and $\Lambda(0)=2$. The equation $\Lambda(\beta) = 0$ has an unique positive solution $\beta_0 \approx 0.437864$. \newline

\noindent If $0 < \beta < \beta_0$ and
\[\frac{{\mathcal D}({\mathcal H})}{\kappa} < \left(\frac{\Lambda(\beta)}{3}\right)^{\frac{3}{2}},\]
then the equation
\[\Lambda(\beta)x - x^3 = \frac{2 {\mathcal D}({\mathcal H})}{\kappa}\]
has two positive solutions $\alpha_1,\alpha_2$ with $0<\alpha_1<\sqrt{\frac{\Lambda(\beta)}{3}}<\alpha_2<\sqrt{\Lambda(\beta)}$.

\begin{lemma} \label{L6.1}
Let $A$ and $B$ be $d\times d$ hermitian matrices. Then, one has
\[ \norm{e^{\mathrm{i}A} - e^{\mathrm{i}B}}_\F \leq \norm{A-B}_\F.\]
\end{lemma}
\begin{proof} Let the eigen-decomposition of $A$ and $B$ be
\[A = U \begin{pmatrix} \lambda_1 & & \\ & \ddots & \\ & & \lambda_d \end{pmatrix} U^\dag,\quad B = V \begin{pmatrix} \mu_1 & & \\ & \ddots & \\ & & \mu_d \end{pmatrix} V^\dag ,\]
where $U$, $V$ are unitary matices and $\lambda_1, \cdots,\lambda_n,\mu_1,\cdots,\mu_n$ are real numbers. \newline

Suppose $\Sigma = \text{diag}(\sigma_1,\cdots,\sigma_d)$ and $\tilde{\Sigma} = \text{diag} (\tilde{\sigma}_1,\cdots,\tilde{\sigma}_d)$ be arbitary diagonal matrices and $X = U\Sigma U^\dag$, $\tilde{X} = V\tilde{\Sigma}V^\dag$. Then we have the following estimate with $W = U^\dag V$:
\[\tr{X\tilde{X}^\dag} = \tr{U \Sigma U^\dag V \tilde{\Sigma}^\dag V^\dag} = \tr{\Sigma W \tilde{\Sigma}^\dag W^\dag} = \sum_{k,l=1}^d \sigma_k \tilde{\sigma}_l^\ast \abs{W_{k,l}}^2 .\]
By the same argument, we have
\begin{align*}
& \tr{e^{\mathrm{i} A} e^{-\mathrm{i} B}} = \sum_{k,l=1}^d e^{\mathrm{i}(\lambda_k - \mu_l)} \abs{W_{k,l}}^2, \quad \tr{e^{\mathrm{i} B} e^{-\mathrm{i} A}} = \sum_{k,l=1}^d e^{\mathrm{i}(\mu_l-\lambda_k)} \abs{W_{k,l}}^2, \\
& \tr{AB^\dag} = \tr{BA^\dag} = \sum_{k,l=1}^d \lambda_k \mu_l \abs{W_{k,l}}^2.
\end{align*}
Therefore, we obtain 
\begin{align*}
\snorm{e^{\mathrm{i}A} - e^{\mathrm{i}B}}_\F^2 &= \tr{(e^{\mathrm{i}A} - e^{\mathrm{i}B})(e^{-\mathrm{i}A} - e^{-\mathrm{i}B})} = \tr{2I - e^{\mathrm{i}A} e^{-\mathrm{i}B} - e^{\mathrm{i} B} e^{-\mathrm{i} A}} \\
&= \sum_{k,l=1}^d (2- e^{\mathrm{i}(\lambda_k-\mu_l)} - e^{-\mathrm{i}(\lambda_k-\mu_l)}) \abs{W_{k,l}}^2 = \sum_{k,l=1}^d 4\sin^2\left(\frac{\lambda_k - \mu_l}{2} \right) \abs{W_{k,l}}^2 \\
&\leq \sum_{k,l=1}^d (\lambda_k - \mu_l)^2 \abs{W_{k,l}}^2 = \sum_{k,l=1}^d (\lambda_k^2 + \mu_l^2 - 2\lambda_k \mu_l) \abs{W_{k,l}}^2 \\
&= \tr{AA^\dag - BA^\dag - AB^\dag + BB^\dag} = \norm{A-B}_\F^2. \qedhere
\end{align*}
\end{proof}

\vspace{0.2cm}

\begin{proposition} \label{P6.1}
Suppose system parameters and initial data satisfy
\[0<\beta =\kappa h < \beta_0 \approx 0.437864,\quad \frac{{\mathcal D}({\mathcal H})}{\kappa} < \left(\frac{\Lambda(\beta)}{3}\right)^{\frac{3}{2}},\quad \mathcal{U}^0\in\mcal{B}(\alpha_2), \]
and let $\{ {\mathcal U}(n) \}$ be a solution to \eqref{D-12-0} with the initial data ${\mathcal U}^0$. Then the following assertions hold:
\begin{enumerate}[label=(\roman*)]
\item ${\mathcal U}(n) \in \mcal{B}(\alpha_2)$ for all $n\geq 0$.
\vspace{0.1cm}
\item For any $\alpha_1<\alpha<\alpha_2$, there exists $n_\alpha$ such that ${\mathcal U}(n)\in\mcal{B}(\alpha)$ for all $n\geq n_\alpha$.
\vspace{0.1cm}
\item If $\,{\mathcal U}^0\in \overline{\mcal{B}(\alpha_1)}$, then ${\mathcal U}(n) \in\overline{\mcal{B}(\alpha_1)}$ for all $n\geq 0$.
\end{enumerate}
\end{proposition}
\begin{proof} Note that the above statements can be followed from the following assertions:
\begin{itemize}
    \item If $\mcal{D}(n)\in[\alpha_1,\alpha_2)$, then $\mcal{D}(n+1)\leq \mcal{D}(n)$.
    \vspace{0.1cm}
    \item If $\alpha\in(\alpha_1,\alpha_2)$ and $\mcal{D}(n)\in[\alpha,\alpha_2)$, then $\mcal{D}(n+1) \leq \mcal{D}(n) - C(\alpha)$ with a positive constant $C(\alpha)$.
    \vspace{0.1cm}
    \item If $\mcal{D}(n)\in[0,\alpha_1]$, then $\mcal{D}(n+1) \leq \alpha_1$.
\end{itemize}
For all $i\in[N]$, we define an intermediate state $V_i(n+1)$ via the relation:
\begin{equation*}
U_i(n+1) = \exp\left(-{\mathrm i}H_i h \right) V_i(n+1). 
\end{equation*}
Then, one has
\[V_i(n+1) = \exp \left(\frac{\kappa h}{2} (U_c(n) U_i^\dag (n) - U_i(n) U_c^\dag (n)) \right) U_i(n) \in \mbf{U}(d), \quad i\in[N],\; n\geq 0.\]
On the other hand, it follows from \eqref{D-12} that 
\begin{equation} \label{D-12-2}
 \lVert V_i(n+ 1) - V_j(n+1) \rVert_\F^2 \leq  \Big( 1 - \beta ( \Lambda(\beta) - (\mcal{D}(n))^2 ) \Big) \lVert U_i(n) - U_j(n) \rVert_\F^2.
\end{equation}
By \eqref{D-12-2}, Lemma \ref{L5.3} and Lemma \ref{L6.1}, one has 
\begin{align*}
\lVert U_i(n & +1) - U_j(n+1) \rVert_\F = \snorm{e^{-{\mathrm i}H_i h} V_i(n+1) - e^{-{\mathrm i}H_j h} V_j(n+1)}_\F\\
&\leq \norm{V_i(n+1) - V_j(n+1)}_\F + \snorm{e^{-{\mathrm i}H_ih} - e^{-{\mathrm i}H_j h}}_\F \\
&\leq \sqrt{1-\beta(\Lambda(\beta) - \mcal{D}(n)^2)} \norm{U_i(n) - U_j(n)}_\F + h \norm{H_i - H_j}_\F \\
&\leq \left(1 - \frac{\beta}{2} (\Lambda(\beta) - \mcal{D}(n)^2) \right) \mcal{D}(n) + h {\mathcal D}({\mathcal H}).
\end{align*}
Therefore, we get
\[\mcal{D}(n+1)\leq \mcal{D}(n) - \frac{\beta}{2}\left( \Lambda(\beta) \mcal{D}(n) - \mcal{D}(n)^3 - \frac{2{\mathcal D}({\mathcal H})}{\kappa} \right).\]
The above inequality directly proves the first two assertions and the last assertion holds since $\frac{\beta}{2} \Lambda(\beta) \leq 0.13205 < 1$ regardless of $\beta$. Let $\alpha_3 <0<\alpha_1<\alpha_2$ be the three real solutions to the cubic equation $x^3 - \Lambda(\beta)x + \frac{2\mcal{D}(\mcal{H})}{\kappa} = 0$. For $0\leq x\leq \alpha_1$, one has
\[\frac{\beta}{2} (x-\alpha_3)(\alpha_2 -x) \leq \frac{\beta}{2} \left(\frac{\alpha_2-\alpha_3}{2} \right)^2\leq \frac{\beta}{2}\Lambda(\beta) <1.\]
Therefore, for $0\leq x\leq \alpha_1$,
\[x + \frac{\beta}{2} \left( x^3 - \Lambda(\beta)x + \frac{2\mcal{D}(\mcal{H})}{\kappa} \right) = x+ \frac{\beta}{2} (x - \alpha_3)(\alpha_2 - x)(\alpha_1 - x) \leq x + (\alpha_1-x) = \alpha_1. \qedhere\]
\end{proof}

\subsubsection{Orbital stability} \label{sec:6.1.2}
In this part, we study asymptotic state-locking for \eqref{D-2-5}. For the zero free flows with $H_i = O$, by Theorem \ref{T5.1}, the relative states tend to the same state. However, for a heterogeneous ensemble with different hamiltonians,  complete state aggregation will not happen. Moreover, the result of Theorem \ref{T5.1} does not tell us whether $U_i(n)$ converges or not, as $n \to \infty$, i.e., we can exclude the possibility in which the common state is time-dependent a priori. \newline

 Let $\{ {\mathcal U}(n) \}$ and $\{ \tilde{\mathcal U}(n) \}$ be two solutions to \eqref{D-2-5}:
\begin{align}
\begin{aligned} \label{D-12-2-1}
 U_i (n+1) &=  \exp(-\mathrm{i} H_i h)  \exp \big(\beta \Delta_i \big) U_i(n), \\
 {\tilde U}_i (n+1) &=  \exp(-\mathrm{i} H_i h)  \exp \big(\beta{\tilde  \Delta}_i \big) {\tilde U}_i(n).
\end{aligned}
\end{align}
Then, we introduce the functional $d({\mathcal U},\tilde{\mathcal{U}})$ measuring the relative position of two configurations $\mathcal{U},\tilde{\mathcal{U}}$:
\[d(\mathcal{U},\tilde{\mathcal{U}}) = \max_{1 \leq i,j \leq N} \snorm{U_i U_j^\dag - \tilde{U}_i \tilde{U}_j^\dag}_\F. \]
It is easy to see that for $d=1$, via ansatz \eqref{New-1}, 
\begin{equation} \label{New-2}
U_i U_j^\dag - \tilde{U}_i \tilde{U}_j^\dag = e^{-{\mathrm i} (\theta_i - \theta_j)} - e^{-{\mathrm i} ({\tilde \theta}_i - {\tilde \theta}_j)}.
\end{equation} 
Hence, the zero convergence of the L.H.S. of \eqref{New-2} is equivalent to the zero convergence of difference between corresponding relative phases $(\theta_i - \theta_j) - ({\tilde \theta}_i - {\tilde \theta}_j)$. In this sense, it is reasonable to guess that the functional $d(\mathcal{U},\tilde{\mathcal{U}})$ measures relative difference between two configurations. 

Note that
\[\snorm{U_i U_j^\dag - \tilde{U}_i \tilde{U}_j^\dag}_\F \leq \snorm{U_i U_j^\dag - I}_\F + \snorm{\tilde{U}_i \tilde{U}_j^\dag - I}_\F \leq \mcal{D}(\mcal{U}) + \mcal{D}(\tilde{\mcal{U}}),\quad i,j\in[N],\]
implies $d(\mcal{U}, \tilde{\mcal{U}}) \leq \mcal{D}(\mcal{U}) + \mcal{D}(\tilde{\mcal{U}})$.

\vspace{0.2cm}

Recall that our goal in this subsection is to look for sufficient condition in which the distance between relative positions $\snorm{U_i(n) U_j^\dag (n) - \tilde{U}_i(n) \tilde{U}_j^\dag (n)}_\F$ is monotonically decreasing, as $n$ increases to infinity. In what follows, unless stated otherwise, the model states $U_i$, $\tilde{U}_i$, and the functionals $U_c$, $\tilde{U}_c$, $\Delta_i$, $\tilde{\Delta}_i$ denote those values at the $n$-th time step. Using \eqref{D-12-2-1}, one has 
{\allowdisplaybreaks \begin{align*}
& \snorm{U_i(n+1) U_j^\dag(n+1) - \tilde{U}_i(n+1) \tilde{U}_j^\dag(n+1)}_\F^2 - \snorm{U_i(n) U_j^\dag (n) - \tilde{U}_i(n) \tilde{U}_j^\dag (n)}_\F^2  \\
&  \hspace{0.5cm} =  \snorm{V_i(n+1) V_j^\dag(n+1) - \tilde{V}_i(n+1) \tilde{V}_j^\dag(n+1)}_\F^2 - \snorm{U_i(n) U_j^\dag (n) - \tilde{U}_i(n) \tilde{U}_j^\dag (n)}_\F^2  \\
& \hspace{0.5cm} = \text{tr}\bigg( \big( U_i U_j^\dag  \tilde{U}_j \tilde{U}_i^\dag \big)(n) + \big( U_j U_i^\dag \tilde{U}_i \tilde{U}_j ^\dag \big)(n) - \big( U_i U_j^\dag \tilde{U}_j \tilde{U}_i^\dag \big)(n+1) - \big( U_j U_i^\dag \tilde{U}_i \tilde{U}_j ^\dag \big) (n+1) \bigg) \\
&  \hspace{0.5cm} = \text{tr}\left( U_i U_j^\dag \tilde{U}_j \tilde{U}_i^\dag - \exp(\beta \Delta_i) U_i U_j^\dag \exp(-\beta\Delta_j) \exp(\beta\tilde{\Delta}_j) \tilde{U}_j \tilde{U}_i^\dag \exp(-\beta\tilde{\Delta}_i) \right. \\
& \hspace{1cm}  + \left. U_j U_i^\dag \tilde{U}_i \tilde{U}_j^\dag - \exp(\beta \Delta_j) U_j U_i^\dag \exp(-\beta\Delta_i) \exp(\beta\tilde{\Delta}_i) \tilde{U}_i \tilde{U}_j^\dag \exp(-\beta\tilde{\Delta}_j) \right)\\
&  \hspace{0.5cm} = \text{tr}\left( U_i U_j^\dag \tilde{U}_j \tilde{U}_i^\dag -\exp(-\beta\tilde{\Delta}_i) \exp(\beta \Delta_i) U_i U_j^\dag \exp(-\beta\Delta_j) \exp(\beta\tilde{\Delta}_j) \tilde{U}_j \tilde{U}_i^\dag \right. \\
& \hspace{1cm}  + \left. U_j U_i^\dag \tilde{U}_i \tilde{U}_j^\dag -\exp(-\beta\tilde{\Delta}_j) \exp(\beta \Delta_j) U_j U_i^\dag \exp(-\beta\Delta_i) \exp(\beta\tilde{\Delta}_i) \tilde{U}_i \tilde{U}_j^\dag \right)\\
& \hspace{0.5cm} = - \beta\mcal{T}_1  - \frac{\beta^2}{2!} \mcal{T}_2 - \frac{\beta^3}{3!} \mcal{T}_3 -\cdots = - \sum_{m=1}^\infty \frac{\beta^m}{m!} \mcal{T}_m, \stepcounter{equation}\tag{\theequation}\label{D-12-1-3}
\end{align*}}
where $\mcal{T}_m$'s are terms independent of $\beta$. For all $m\geq 1$, we define
\begin{align*}
\mcal{I}_m^+ &= \sum_{k=0}^m \frac{m!}{k!(m-k)!} (-\tilde{\Delta}_i)^k \Delta_i^{m-k},\quad \mcal{I}_m^- = \sum_{k=0}^m \frac{m!}{k!(m-k)!} (-\Delta_i)^k \tilde{\Delta}_i^{m-k}, \\
\mcal{J}_m^+ &= \sum_{k=0}^m \frac{m!}{k!(m-k)!} (-\tilde{\Delta}_j)^k \Delta_j^{m-k},\quad \mcal{J}_m^- = \sum_{k=0}^m \frac{m!}{k!(m-k)!} (-\Delta_j)^k \tilde{\Delta}_j^{m-k}l.
\end{align*}
For convenience, we set 
\[ \mcal{I}_0^+ = \mcal{I}_0^- = \mcal{J}_0^+ = \mcal{J}_0^- = I. \]
As formula \eqref{D-8}, expanding the matrix exponentials results in the following identities:
\begin{align*}
\exp(-\beta\tilde{\Delta}_i)\exp(\beta\Delta_i) &= \sum_{m=0}^\infty \frac{\beta^m}{m!} \mcal{I}_m^+,\quad \exp(-\beta\Delta_i)\exp(\beta\tilde{\Delta}_i) = \sum_{m=0}^\infty \frac{\beta^m}{m!} \mcal{I}_m^-, \\
\exp(-\beta\tilde{\Delta}_j)\exp(\beta\Delta_j) &= \sum_{m=0}^\infty \frac{\beta^m}{m!} \mcal{J}_m^+,\quad \exp(-\beta\Delta_j)\exp(\beta\tilde{\Delta}_j) = \sum_{m=0}^\infty \frac{\beta^m}{m!} \mcal{J}_m^-.
\end{align*}
Thus, we have
\begin{equation} \label{D-14}
\mcal{T}_m = \sum_{k=0}^m \frac{m!}{k!(m-k)!} \bigg\{\underbrace{\tr{\mcal{I}_k^+ U_i U_j^\dag \mcal{J}_{m-k}^- \tilde{U}_j \tilde{U}_i^\dag} + \tr{\mcal{J}_{m-k}^+ U_j U_i^\dag \mcal{I}_{k}^- \tilde{U}_i \tilde{U}_j^\dag}}_{\mcal{M}_{k,m-k}} \bigg\}.
\end{equation}
Note that
{\allowdisplaybreaks \begin{align*}
\mcal{I}_{m+1}^+ &= \sum_{k=0}^{m+1} \frac{(m+1)!}{k!(m+1-k)!} (-\tilde{\Delta}_i)^k \Delta_i^{m+1-k} \\
&= \sum_{k=0}^{m+1} \left( \frac{m!}{k!(m-k)!} + \frac{m!}{(k-1)! (m+1-k)!} \right) (-\tilde{\Delta}_i)^k \Delta_i^{m+1-k} \\
&= \sum_{k=0}^{m} \frac{m!}{k! (m-k)!} (-\tilde{\Delta}_i)^k \Delta_i^{m+1-k} + \sum_{k=0}^{m} \frac{m!}{k! (m-k)!} (-\tilde{\Delta}_i)^{k+1} \Delta_i^{m-k}\\
&= \sum_{k=0}^{m} \frac{m!}{k! (m-k)!} (-\tilde{\Delta}_i)^k (\Delta_i - \tilde{\Delta}_i) \Delta_i^{m-k}.
\end{align*}}
We have similar formulas for $\mcal{I}_m^-$, $\mcal{J}_m^+$, and $\mcal{J}_m^-$. By Lemma \ref{L5.1}, Lemma \ref{L5.2} and Lemma \ref{L5.4}, we have
\[\norm{\mcal{I}_{m+1}^+}_\F \leq \sum_{k=0}^{m} \frac{m!}{k!(m-k)!} \snorm{\Delta_i - \tilde{\Delta}_i}_\F = 2^{m} \snorm{\Delta_i - \tilde{\Delta}_i}_\F \]
and
\[\snorm{\Delta_i - \tilde{\Delta}_i}_\F \leq \frac{1}{2N} \sum_{k=1}^N \big(\snorm{U_k U_i^\dag - \tilde{U}_k \tilde{U}_i^\dag}_\F + \snorm{U_i U_k^\dag - \tilde{U}_i \tilde{U}_k^\dag}_\F \big) \leq d(\mathcal{U},\tilde{\mathcal{U}}).\]
Therefore, we have
\begin{equation} \label{D-13}
\norm{\mcal{I}_{m}^+}_\F, \norm{\mcal{I}_{m}^-}_\F, \norm{\mcal{J}_{m}^+}_\F, \norm{\mcal{J}_{m}^-}_\F \leq 2^{m-1} d(\mathcal{U},\tilde{\mathcal{U}}),\quad m\geq 1.
\end{equation}

\vspace{0.1cm}

\noindent $\bullet$~(Estimate of $\mcal{T}_1$):~The first order term $\mcal{T}_1$ is studied using the argument \cite{H-R} (Appendix 2) for the Lohe matrix model:
{\allowdisplaybreaks \begin{align}
\begin{aligned} \label{D-13-1}
\mcal{T}_1 &= \text{tr}\left((\Delta_i - \tilde{\Delta}_i) U_i U_j^\dag \tilde{U}_j \tilde{U}_i^\dag - U_i U_j^\dag (\Delta_j - \tilde{\Delta}_j) \tilde{U}_j \tilde{U}_i^\dag \right. \\
& \hspace{0.8cm} + \left. (\Delta_j - \tilde{\Delta}_j) U_j U_i^\dag \tilde{U}_i \tilde{U}_j^\dag - U_j U_i^\dag (\Delta_i - \tilde{\Delta}_i) \tilde{U}_i \tilde{U}_j^\dag \right) \\
&\geq 2\snorm{U_i U_j^\dag - \tilde{U}_i \tilde{U}_j^\dag}_\F^2 - (4\mcal{D}(\mathcal{U}) + 2\mcal{D}(\tilde{\mathcal{U}})) d(\mathcal{U},\tilde{\mathcal{U}})^2.
\end{aligned}
\end{align}}
In next lemma, we estimate the terms ${\mathcal T}_m$ for $m\geq 2$.
\begin{lemma} \label{L6.2} 
Let $\{{\mathcal U}(n)\}$ be a solution to \eqref{D-12-0} with the initial data ${\mathcal U}^0$. Then, for $m\geq 1$, the following inequality holds:
\[\abs{\mcal{T}_m} \leq \left(2^{2m-1} + m\cdot 2^m\right) d(\mathcal{U},\tilde{\mathcal{U}})^2\]
\end{lemma}
\begin{proof} 
We estimate each summand $\mcal{M}_{k,m-k}$ in the formula \eqref{D-14}. For $1\leq k\leq m-1$, we have the following by \eqref{D-13}:
\begin{align*}
\abs{\mcal{M}_{k,m-k}} &= \abs{ \tr{\mcal{I}_k^+ U_i U_j^\dag \mcal{J}_{m-k}^- \tilde{U}_j \tilde{U}_i^\dag} + \tr{\mcal{J}_{m-k}^+ U_i U_j^\dag \mcal{I}_{k}^- \tilde{U}_j \tilde{U}_i^\dag}} \\
&\leq \norm{\mcal{I}_k^+}_\F \norm{\mcal{J}_{m-k}^-}_\F + \norm{\mcal{J}_{m-k}^+}_\F \norm{\mcal{I}_{m-k}^-}_\F \\
&\leq 2^{k-1} 2^{m-k-1} d(\mathcal{U},\tilde{\mathcal{U}})^2 + 2^{m-k-1} 2^{k-1} d(\mathcal{U},\tilde{\mathcal{U}})^2 = 2^{m-1} d(\mathcal{U},\tilde{\mathcal{U}})^2.
\end{align*}
For $\mcal{M}_{m,0}$, we have
{\allowdisplaybreaks \begin{align*}
\mcal{M}_{m,0} &= \tr{\mcal{I}_m^+ U_i U_j^\dag \tilde{U}_j \tilde{U}_i^\dag} + \tr{U_j U_i^\dag \mcal{I}_{m}^- \tilde{U}_i \tilde{U}_j^\dag} \\
&= \tr{\mcal{I}_m^+ U_i U_j^\dag \tilde{U}_j \tilde{U}_i^\dag} + \tr{\mcal{I}_{m}^- \tilde{U}_i \tilde{U}_j^\dag U_j U_i^\dag} \\
&= \tr{\sum_{k=0}^{m-1} \frac{(m-1)!}{k! (m-1-k)!} (-\tilde{\Delta}_i)^k (\Delta_i - \tilde{\Delta}_i) \Delta_i^{m-1-k} U_i U_j^\dag \tilde{U}_j \tilde{U}_i^\dag } \\
& \hspace{0.5cm} - \tr{\sum_{k=0}^{m-1} \frac{(m-1)!}{k! (m-1-k)!} (-\Delta_i)^k (\Delta_i - \tilde{\Delta}_i) \tilde{\Delta}_i^{m-1-k} \tilde{U}_i \tilde{U}_j^\dag U_j U_i^\dag} \\
& = \sum_{k=0}^{m-1} (-1)^k \frac{(m-1)!}{k! (m-1-k)!} \tr{(\Delta_i - \tilde{\Delta}_i) \left( \Delta_i^{m-1-k} U_i U_j^\dag \tilde{U}_j \tilde{U}_i^\dag \tilde{\Delta}_i^k - \tilde{\Delta}_i^{m-1-k} \tilde{U}_i \tilde{U}_j^\dag U_j U_i^\dag \Delta_i^k \right)}.
\end{align*}}
Hence, we have
\begin{align*}
& |\mcal{M}_{m,0}| = \abs{\tr{\mcal{I}_m^+ U_i U_j^\dag \tilde{U}_j \tilde{U}_i^\dag} + \tr{U_j U_i^\dag \mcal{I}_{m}^- \tilde{U}_i \tilde{U}_j^\dag}} \\
& \quad \leq \sum_{k=0}^{m-1} \frac{(m-1)!}{k! (m-1-k)!} \snorm{\Delta_i - \tilde{\Delta}_i}_\F \norm{\Delta_i^{m-1-k} U_i U_j^\dag \tilde{U}_j \tilde{U}_i^\dag \tilde{\Delta}_i^k - \tilde{\Delta}_i^{m-1-k} \tilde{U}_i \tilde{U}_j^\dag U_j U_i^\dag \Delta_i^k}_\F \\
& \quad \leq \sum_{k=0}^{m-1} \frac{(m-1)!}{k! (m-1-k)!} d(\mathcal{U},\tilde{\mathcal{U}}) \cdot (m+1) d(\mathcal{U},\tilde{\mathcal{U}}) = (m+1) 2^{m-1} d(\mathcal{U},\tilde{\mathcal{U}})^2.
\end{align*}
In the third line of the above inequality, we used Lemma \ref{L5.3} and Lemma \ref{L5.4} to get
\begin{align*}
& \norm{\Delta_i^{m-1-k} U_i U_j^\dag \tilde{U}_j \tilde{U}_i^\dag \tilde{\Delta}_i^k - \tilde{\Delta}_i^{m-1-k} \tilde{U}_i \tilde{U}_j^\dag U_j U_i^\dag \Delta_i^k}_\F \\
& \hspace{0.2cm} \leq (m-1-k) \snorm{\Delta_i - \tilde{\Delta}_i}_\F + \snorm{U_i U_j^\dag - \tilde{U}_i \tilde{U}_j^\dag}_\F + \snorm{\tilde{U}_j \tilde{U}_i^\dag - U_j U_i^\dag}_\F + k\snorm{\tilde{\Delta}_i - \Delta_i}_\F \\
& \hspace{0.2cm}  \leq (m+1)d(\mathcal{U},\tilde{\mathcal{U}}).
\end{align*}
By the same argument, we have
\[\abs{\mcal{M}_{0,m}} \leq (m+1)2^{m-1} d(\mathcal{U},\tilde{\mathcal{U}})^2. \]
Summing up, we get \vspace{-0.2cm}
\begin{align*}
\abs{\mcal{T}_m} & \leq \abs{\mcal{M}_{m,0}} + \abs{\mcal{M}_{0,m}} + \sum_{k=1}^{m-1} \frac{m!}{k!(m-k)!} \abs{\mcal{M}_{k,m-k}} \\
&\leq (m+1)2^m d(\mathcal{U},\tilde{\mathcal{U}})^2 + (2^m -2) 2^{m-1} d(\mathcal{U},\tilde{\mathcal{U}})^2 \\
&= (2^{2m-1} + m\cdot 2^m ) d(\mathcal{U},\tilde{\mathcal{U}})^2. \qedhere
\end{align*}
\end{proof}

\vspace{0.2cm}

Let $M(\beta) = \frac{1}{6}\left(6 - 2e^{2\beta} - \frac{e^{4\beta}-1}{2\beta}\right)$. Then $M(\beta)$ is a decreasing function of $\beta>0$ and $M(0)=2$. The equation $M(\beta)=0$ has a unique positive solution $\beta_1\approx 0.196302 < \beta_0\approx 0.437864$. We also note that $M(\beta) < \sqrt{\frac{\Lambda(\beta)}{3}} < \alpha_2$ for $0<\beta < \beta_1$.
\begin{theorem} \label{T6.1}
Suppose system parameters and initial data satisfy 
\[  0<\beta = \kappa h < \beta_1\approx 0.196302, \quad \mathcal{U}^0\in \overline{\mcal{B}(\alpha)},\quad 0 < \alpha < M(\beta),\]
and let $\{\mathcal{U}(n)\}$ and $\{\tilde{\mathcal{U}}(n)\}$ be solutions to \eqref{D-12-0} corresponding to the initial data $\mathcal{U}^0$ and $\tilde{\mathcal{U}}^0$ respectively,
and $\, \mathcal{U}(n), \tilde{\mathcal{U}}(n)\in\overline{\mcal{B}(\alpha)}$ for all $n\geq 0$. Then the following assertions hold:
\begin{enumerate}[label=(\roman*)]
\item The relative positions synchronize exponentially fast, i.e., there exists a constant $0<C<1$ such that 
\[d(\mathcal{U}(n), \tilde{\mathcal{U}}(n)) \leq C^n d(\mathcal{U}^0,\tilde{\mathcal{U}}^0),\quad n\geq 0\]
\item There exists a unitary matrix $L^\infty \in \mathbf{U}(d)$ independent of $i$ such that
\[\lim_{n\to\infty} U_i(n) \tilde{U}_i^\dag(n) = L^\infty.\]
\end{enumerate}
\end{theorem}
\begin{proof} 
(i)~We use \eqref{D-12-1-3} and the estimates of $\mcal{T}_m$'s (Lemma \ref{L6.2}) to get 
\begin{align}
\begin{aligned} \label{D-14-1}
&\snorm{U_i(n+1) U_j^\dag (n+1) - \tilde{U}_i(n+1) \tilde{U}_j^\dag (n+1)}_\F^2  \\
& \hspace{0.5cm} = \snorm{U_i(n) U_j^\dag (n) - \tilde{U}_i(n) \tilde{U}_j^\dag (n)}_\F^2 - \sum_{m=1}^\infty \frac{\beta^m}{m!} \mcal{T}_m \\
& \hspace{0.5cm} \leq (1-2\beta) \snorm{U_i(n) U_j^\dag (n) - \tilde{U}_i(n) \tilde{U}_j^\dag (n)}_\F^2 + \left(6\alpha \beta + \sum_{m=2}^\infty \frac{\beta^m}{m!}(2^{2m-1} + m\cdot 2^m) \right) d(\mathcal{U},\tilde{\mathcal
{U}})^2 \\
& \hspace{0.5cm} \leq \left( 1- \beta \left(6- 2 e^{2\beta} - \frac{e^{4\beta}-1}{2\beta} - 6\alpha \right) \right) d(\mathcal{U}, \tilde{\mathcal{U}})^2.
\end{aligned}
\end{align}
Under the given conditions, the assertion (a) holds with $C = \sqrt{1-6\beta(M(\beta)-\alpha)}$.\\

\noindent (ii) By Lemma \ref{L6.1}, we get
\begin{align*}
\lVert U_i(n+1) & \tilde{U}_i^\dag (n+1) - U_i(n) \tilde{U}_i^\dag (n) \rVert_\F = \snorm{U_i(n+1) U_i^\dag (n) - \tilde{U}_i(n+1) \tilde{U}_i^\dag (n)}_\F \\
&= \snorm{\exp(\beta \Delta_i(n)) - \exp(\beta \tilde{\Delta}_i(n))}_\F \\
&\leq \beta \snorm{\Delta_i(n) - \tilde{\Delta}_i(n)}_\F \leq \beta d(\mathcal{U}(n), \tilde{\mathcal{U}}(n)) \leq \beta C^n d(\mathcal{U}^0, \tilde{\mathcal{U}}^0)
\end{align*}
for all $n\geq 0$. Thus, one has
\[\lVert U_i(n+m) \tilde{U}_i^\dag (n+m)- U_i(n) \tilde{U}_i^\dag (n) \rVert_\F < \frac{\beta C^n}{1-C} d(\mathcal{U}^0, \tilde{\mathcal{U}}^0) \]
for all $n,m\geq 0$. This implies that $\{U_i(n) \tilde{U}_i^\dag (n)\}_{n \geq 0}$ is a Cauchy sequence. It follows that the following limit exists.
\[L_i^\infty := \lim_{n\to\infty} U_i(n) \tilde{U}_i^\dag (n).\]
According to (a), we have $d(\mcal{U}(n), \tilde{\mcal{U}}(n)) \to 0$ as $n\to \infty$. This implies 
\[\lim_{n\to\infty} \snorm{U_i(n) \tilde{U}_i^\dag(n) - U_j(n) \tilde{U}_j^\dag (n)}_\F = \lim_{n\to\infty} \snorm{U_i(n) U_j^\dag (n) - \tilde{U}_i(n) \tilde{U}_j^\dag (n)}_\F = 0\]
for all $i,j\in[N]$. Therefore, $L^\infty_i = L^\infty_j$ for all $i,j\in[N]$ and thus
the assertion (ii) holds.
\end{proof}

\begin{remark}
If $\beta = 0.1$, the condition above for the coupling strength $\kappa$ is equivalent to
\[\kappa > 6.67134 \cdot {\mathcal D}({\mathcal H}).\]

\end{remark}

\subsubsection{Asymptotic state-locking} \label{sec:6.1.3}
In this part, we study a sufficient framework leading to state-locking as follows.
\begin{theorem} \label{T6.2}
Suppose system parameters and initial data satisfy
\[0<\beta=\kappa h < \beta_1\approx 0.196302,\quad \frac{\mcal{D}(\mcal{H})}{\kappa}< \frac{1}{2}\left(\Lambda(\beta) M(\beta) - M(\beta)^3 \right),\quad \mathcal{U}^0\in \mcal{B}(\alpha_2), \]
and let $\{\mcal{U}(n)\}$ be a solution to \eqref{D-12-0} with initial data $\mathcal{U}^0$. Then, the following assertions hold. 
\begin{enumerate}[label=(\roman*)]
\item
The discrete-time Lohe flow $\{\mathcal{U}(n)\}_{n\geq 0}$ achieves asymptotic state-locking:
\[ \lim_{n\to\infty} U_i(n) U_j^\dag (n)\]
converges exponentially fast for all $i, j \in [N]$.

\vspace{0.2cm}

\item
Let $\{\mathcal{U}^h(n)\}_{n\geq 0}$ be a solution to \eqref{D-12-0} with time step $h$ and initial data $\mathcal{U}^0$. The uniform-in-time convergence of the discrete solution to the continuous solution occurs:
\begin{equation} \label{D-15}
\limsup_{h\to 0} \sup_{0\leq n<\infty} d(\mathcal{U}^h(n),\mathcal{U}(nh))= 0.
\end{equation}
\end{enumerate}
\end{theorem}
\begin{proof}
\noindent (i)~There exists $\alpha \in (\alpha_1,M(\beta))$. According to Proposition \ref{P6.1}, there exists $n_\alpha\geq 0$ such that 
\[ \mcal{U}(n)\in\mcal{B}(\alpha) \quad \mbox{for all $n\geq n_\alpha$}. \]
So we may assume that $\mathcal{U}^0 \in\mcal{B}(\alpha)$. This allows us to apply Theorem \ref{T6.1}. The time-shifted sequence $\{\mathcal{U}(n+1)\}_{n\geq 0}$ is a solution to \eqref{D-12-0} with the time-shifted initial data $\mathcal{U}(1)$. Hence the exponential stability estimate in Theorem \ref{T6.1} implies 
\[\snorm{U_i (n+1) U_j^\dag (n+1) - U_i(n) U_j^\dag (n)}_\F \leq C^n d(\mathcal{U}(1),\mathcal{U}(0)),\quad n\geq 0,\]
for some constant $0<C<1$. The above estimate shows that the sequence $\{U_i(n) U_j^\dag (n)\}_{n\geq 0}$ is a Cauchy sequence.
\[\snorm{U_i (n+m) U_j^\dag (n+m) - U_i(n) U_j^\dag (n)}_\F < \frac{C^n}{1-C} d(\mathcal{U}(1),\mathcal{U}(0)),\quad n,m\geq 0.\]
Hence the sequence $\{U_i(n) U_j^\dag (n)\}_{n\geq 0}$ converges for all $i,j\in[N]$ and taking $m\to\infty$ to the above inequality shows that the convergence is exponential. \newline

\noindent (ii)~By the exponential stability estimate of the discrete Lohe matrix model (Theorem \ref{T6.1}) and the Lohe matrix model \cite{H-R}, we get
\begin{align*}
d(\mathcal{U}^h(n), \mathcal{U}(nh)) &\leq \sum_{m=1}^n d(\mathcal{U}^h(m-1), \mathcal{U}^h(m)) + \sum_{m=1}^n d(\mathcal{U}((m-1)h),\mathcal{U}(mh)) \\
&< C\left(d(\mathcal{U}^h(0), \mathcal{U}^h(1)) + d(\mathcal{U}(0), \mathcal{U}(h)) \right) \leq C'h,
\end{align*}
where $C,C'$ are positive constants. This proves uniform-in-time convergence \eqref{D-15}.
\end{proof}
\begin{remark}
Note that the discrete Lohe model exhibits asymptotic state-locking under sufficient regularity conditions. Moreover, the discrete Lohe matrix model  B \eqref{D-12-0} converges to the continuous model in the relative position sense.
\end{remark}

\subsection{The discrete Lohe matrix model C}  \label{sec:6.2}
In this subsection, we study asymptotic dynamics of the Cauchy problem to the discrete Lohe matrix model C:
\begin{equation}
\begin{cases} \label{D-16}
\displaystyle U_i (n+1) \\
\displaystyle =  \exp\left(\frac{-\mathrm{i} H_i h}{2} \right)  \exp \left(\frac{\kappa h}{2}(U_c(n) U_i^\dag (n) - U_i(n) U_c^\dag (n))\right) \exp\left(\frac{-\mathrm{i} H_i h}{2} \right) U_i(n),~~ n \geq 0, \\
U_i(0) = U^0_i, \quad i \in [N].
\end{cases}
\end{equation}

\noindent Throughout this subsection, we assume
\[\sum_{k=1}^N H_k = O,\]
so that 
\[ \norm{H_i}_\F \leq \mcal{D}(\mcal{H}) \quad \mbox{for all $i\in[N]$.} \]
\vspace{0.2cm}

\noindent We define an intermediate state $\mcal{V}(n) := (V_1(n),\cdots,V_N(n))$ via the following relations: 
\[V_i(n) := \exp\left(\frac{-\mathrm{i} H_i h}{2} \right) U_i(n), \quad n\geq 0, \quad i\in[N], \]
and some functionals:
\[V_c := \frac{1}{N} \sum_{k=1}^N V_k,\quad \vDelta_i := \frac{1}{2}(V_c V_i^\dag - V_i V_c^\dag), \quad i\in[N].\]
Then, $\mcal{V}(n)$ satisfies
\begin{equation} \label{D-16-0}
V_i(n+1) = \exp\left(-\mathrm{i} H_i h \right) \exp(\beta \Delta_i) V_i(n), \quad n\geq 0, \quad i\in[N].
\end{equation}
By Lemma \ref{L5.1} and Lemma \ref{L6.1}, we have
\begin{equation} \label{D-16-1-1}
\norm{V_i - U_i}_\F = \norm{\exp\left( \frac{-\mathrm{i} H_i h}{2} \right) - I}_\F \leq \frac{h}{2} \norm{H_i}_\F \leq \frac{h}{2} \mcal{D}(\mcal{H}), \quad i\in[N],
\end{equation}
and
\[\norm{V_i - V_j}_\F \leq \norm{U_i - U_j}_\F + \norm{\exp\left(\frac{-\mathrm{i} H_i h}{2} \right) - \exp\left(\frac{-\mathrm{i} H_j h}{2} \right)}_\F \leq \mcal{D}(\mcal{U}) + \frac{h}{2} \mcal{D}(\mcal{H}),\quad i,j\in[N].\]
Thus, one has
\[  \mcal{D}(\mcal{V})\leq \mcal{D}(\mcal{U}) + \frac{h}{2} \mcal{D}(\mcal{H}). \]
Similarly, we get 
\[ \mcal{D}(\mcal{U})\leq \mcal{D}(\mcal{V}) + \frac{h}{2} \mcal{D}(\mcal{H}) \]
and thus
\begin{equation} \label{D-16-1}
\abs{\mcal{D}(\mcal{U}) - \mcal{D}(\mcal{V})} \leq \frac{h}{2} \mcal{D}(\mcal{H}).
\end{equation}
Again by Lemma \ref{L5.1}, we have
\[\snorm{V_i V_j^\dag - \tilde{V}_i \tilde{V}_j^\dag}_\F = \snorm{U_i U_j^\dag - \tilde{U}_i \tilde{U}_j^\dag}_\F ,\quad i,j\in[N]. \]
Thus, one has 
\begin{equation} \label{D-16-2}
d(\mcal{V}, \tilde{\mcal{V}}) = d(\mcal{U}, \tilde{\mcal{U}}).
\end{equation}

\subsubsection{Orbital stability} In this part, we study time-evolution of state diameter.  Let $\{\mcal{U}(n)\}$ and $\{\tilde{\mcal{U}}(n)\}$ be solutions to the discrete Lohe matrix model C \eqref{D-16}. Since \eqref{D-16-0} is similar to the discrete Lohe matrix model B, we expect some estimates of
\[ \norm{U_i - U_j}_\F^2 (n+1) - \norm{U_i - U_j}_\F^2 (n), \quad \snorm{U_i U_j^\dag - \tilde{U}_i \tilde{U}_j^\dag}_\F^2 (n+1) - \snorm{U_i U_j^\dag - \tilde{U}_i \tilde{U}_j^\dag}_\F^2 (n),
\]
which are similar to those of the discrete Lohe matrix model B which was extensively studied in the previous subsection. Again, we define an intermediate state $W_i(n+1)$ via the relation:
\[V_i(n+1) = \exp\left(-\mathrm{i} H_i h \right) W_i(n+1), \quad i\in[N],\]
so that it satisfies
\[W_i(n+1) = \exp(\beta \Delta_i) V_i(n) \in \mathbf{U}(d),\quad i\in[N].\]
Then, we get the following analogue of \eqref{D-9}:
\begin{equation} \label{D-16-3}
\norm{W_i(n+1) - W_j(n+1)}_\F^2 - \norm{V_i(n) - V_j(n)}_\F^2 = -\sum_{m=1}^\infty \frac{\beta^m}{m!} \widehat{\mcal{A}}_m
\end{equation}
where, for $m\geq 1$,
\begin{align*}
\widehat{\mcal{A}}_m &= \text{tr}\Bigg\{ \bigg( \sum_{k=0}^m \frac{m!}{k!(m-k)!} (-1)^k \Delta_j^k \Delta_i^{m-k} \bigg) V_i V_j^\dag \Bigg\} + \text{tr} \Bigg\{ \bigg( \sum_{k=0}^m \frac{m!}{k!(m-k)!} (-1)^k \Delta_i^k \Delta_j^{m-k} \bigg) V_j V_i^\dag\Bigg\}.
\end{align*}
By Lemma \ref{L5.3} and Lemma \ref{L6.1}, one has
\begin{align}
\begin{aligned} \label{D-16-3-1}
& \snorm{V_i(n+1) - V_j(n+1)}_\F \\
& \hspace{1cm} \leq \snorm{W_i(n+1) - W_j(n+1)}_\F + \snorm{\exp(-\mathrm{i} H_i h) - \exp(-\mathrm{i} H_j h)}_\F \\
& \hspace{1cm} \leq \snorm{W_i(n+1) - W_j(n+1)}_\F + h \mcal{D}(\mcal{H}),
\end{aligned}
\end{align}
for $i,j\in[N]$. First, we use the same arguments as in the proof of Lemma \ref{L5.6} to get
\begin{equation} \label{D-16-4}
\big| \widehat{\mcal{A}}_{m+1} \big| \leq 2^m \cdot \mcal{D}(\mcal{U}) (m\mcal{D}(\mcal{U}) + 2\mcal{D}(\mcal{V})),\quad m\geq 1.
\end{equation}

\noindent Next, we consider
\[\mcal{A}_1^{\mcal{V}} := \tr{(\vDelta_i - \vDelta_j) (V_i V_j^\dag - V_j V_i^\dag)}.\]
By Lemma \ref{L5.3} and \eqref{D-16-1-1}, one has
\begin{align*}
\begin{aligned}
\snorm{\vDelta_i - \Delta_i}_\F & \leq \frac{1}{2N} \sum_{k=1}^N \Big( \snorm{V_k V_i^\dag - U_k U_i^\dag}_\F + \snorm{V_i V_k^\dag - U_i U_k^\dag}_\F \Big) \\
& \leq \frac{1}{2N} \sum_{k=1}^N \Big( \snorm{V_k - U_k}_\F + \snorm{V_i^\dag - U_i^\dag}_\F + \snorm{V_i - U_i}_\F + \snorm{V_k^\dag - U_k^\dag}_\F \Big) \\
& \leq \frac{1}{2N} \cdot N \cdot \frac{h}{2} \mcal{D}(\mcal{H}) \cdot 4 = h\mcal{D}(\mcal{H}), \quad i \in [N].
\end{aligned}
\end{align*}
Again by Lemma \ref{L5.3}, one has
\[\snorm{V_i V_j^\dag - V_j V_i^\dag }_\F \leq \snorm{V_i - V_j}_\F + \snorm{V_j^\dag - V_i^\dag}_\F \leq 2\mcal{D}(\mcal{V}),\quad i,j\in[N].\]
Combining above inequalities by Lemma 4.1, we get
\begin{align} \label{D-16-5}
\begin{aligned}
\big| \mcal{A}_1^{\mcal{V}} - \widehat{\mcal{A}}_1 \big| &\leq \snorm{(\vDelta_i - \Delta_i) - (\vDelta_j - \Delta_j)}_\F \snorm{V_i V_j^\dag - V_j V_i^\dag}_\F \\
&\leq 2\cdot h \mcal{D}(\mcal{H}) \cdot 2 \mcal{D}(\mcal{V}) = 4h \mcal{D}(\mcal{H}) \mcal{D}(\mcal{V}).
\end{aligned}
\end{align}
Therefore, by Lemma \ref{L5.5} and \eqref{D-16-5}, we have
\begin{equation} \label{D-16-6}
\widehat{\mcal{A}}_1 \geq \mcal{A}_1^{\mcal{V}} - \big| \mcal{A}_1^{\mcal{V}} - \widehat{\mcal{A}}_1 \big| \geq 2\snorm{V_i - V_j}_\F^2 - \mcal{D}(\mcal{V})^2 \snorm{V_i - V_j}_\F^2 - 4h \mcal{D}(\mcal{H})\mcal{D}(\mcal{V}) .
\end{equation}
Finally, we use \eqref{D-16-3} and the estimates of $\widehat{\mcal{A}}_m$'s \eqref{D-16-4}, \eqref{D-16-6} to get:
\begin{align} 
\begin{aligned} \label{D-16-7}
& \lVert W_i(n+ 1) - W_j(n+1) \rVert_\F^2 - \lVert V_i(n) - V_j(n) \rVert_\F^2 = -\beta \widehat{\mcal{A}}_1 - \sum_{m=2}^\infty \frac{\beta^m}{m!} \widehat{\mcal{A}}_m \\
& \begin{array}{l} \displaystyle \hspace{0.5cm} \leq - 2\beta \norm{V_i - V_j}_\F^2 + \beta\mcal{D}(\mcal{V})^4 + 4\beta h \mcal{D}(\mcal{H}) \mcal{D}(\mcal{V}) \\
\displaystyle \hspace{0.7cm}  + \sum_{m=2}^\infty \frac{\beta^m}{m!} \cdot 2^{m-1} \big((m-1) \mcal{D}(\mcal{U})^2 + 2 \mcal{D}(\mcal{U}) \mcal{D}(\mcal{V}) \big) \end{array} \qquad \text{ at }n \text{-th time step} \\
& \hspace{0.5cm} \leq -2\beta \norm{V_i - V_j}_\F^2 + \beta \mcal{D}(\mcal{V})^4 + 4 \beta^2 \cdot \frac{\mcal{D}(\mcal{H})}{\kappa} \cdot \mcal{D}(\mcal{V}) + \mcal{O}(\beta^2) \big( \mcal{D}(\mcal{U})^2 + \mcal{D}(\mcal{U}) \mcal{D}(\mcal{V}) \big),
\end{aligned}
\end{align}
where $\mcal{O}(g(\beta))$ denotes some positive smooth function $f(\beta)$ of $\beta$ with $f(\beta)/g(\beta)$ being bounded near $\beta = 0$. \newline

For $0<\beta<\frac{1}{2}$, the term $\norm{V_i(n)-V_j(n)}_\F$ in \eqref{D-16-7} can be replaced to $\mcal{D}(\mcal{V})$.

\begin{proposition} \label{P6.2}
Suppose system parameters satisfy
\[\quad 0 < \beta = \kappa h < \beta^\ast , \quad \frac{\mcal{D}(\mcal{H})}{\kappa} \leq C \beta^{1+\epsilon}, \quad \sum_{k=1}^N H_k = O, \]
where $C$ and $\epsilon$ are positive constants. Then, there exist $\beta^\ast = \beta^\ast(C,\epsilon) > 0$ and $\alpha = \alpha(\beta, \epsilon) > 0$ such that $\,\mcal{V}(n) \in \mcal{B}(\alpha)$ whenever $\{\mcal{U}(n)\}$ is a solution to \eqref{D-16} with an inital data $\mcal{U}^0$ with $\mcal{V}^0 \in \mcal{B}(\alpha)$.
\end{proposition}
\begin{proof} We set
\[\delta:= \frac{\mcal{D}(\mcal{H})}{2\kappa}, \quad \alpha = \beta^{\epsilon/2}. \]
We claim that
\[\mcal{D}(\mcal{V}(n))<\alpha \implies \mcal{D}(\mcal{V}(n+1)) < \alpha,\quad n\geq 0,\]
if $\beta$ is sufficiently small. \newline

We rewrite \eqref{D-16-1} and \eqref{D-16-3-1} as follows:
\[\norm{V_i(n+1) - V_j(n+1)}_\F \leq \norm{W_i(n+1) - W_j(n+1)}_\F + 2 \delta, \quad \abs{\mcal{D}(\mcal{U}) - \mcal{D}(\mcal{V})} \leq \delta .\]
Then $\mcal{D}(\mcal{V}(n))<\alpha$ and \eqref{D-16-7} yield
\begin{align}
\begin{aligned} \label{D-16-8}
& \norm{W_i (n+1) - W_j (n+1) }_\F^2 \\
& \hspace{1cm} \leq (1-2\beta) \mcal{D}(\mcal{V})^2 + \beta \mcal{D}(\mcal{V})^4 + \mcal{O}(\beta^2) \mcal{D}(\mcal{V})^2 + \mcal{O}(\beta^2) \delta \mcal{D}(\mcal{V}) + \mcal{O}(\beta^2) \delta^2 \\
& \hspace{1cm} < (1-2\beta) \alpha^2 + \beta \alpha^4 + \mcal{O}(\beta^2) \alpha^2  + \mcal{O}(\beta^2) \delta \alpha + \mcal{O}(\beta) \delta^2 .
\end{aligned}
\end{align}
We get the following from \eqref{D-16-8} and the concavity of the square root: For all $i,j\in[N]$,
\begin{align*}
\begin{aligned}
& \norm{V_i(n+1) - V_j(n+1)}_\F \leq \norm{W_i(n+1) - W_j(n+1)}_\F + 2\delta \\
& \hspace{0.5cm} < \sqrt{\alpha^2 + \beta\left( -2\alpha^2 + \alpha^4 + \mcal{O}(\beta) \alpha^2 + \mcal{O}(\beta) \delta \alpha + \mcal{O}(\beta) \delta^2 \right)} + 2\delta \\
& \hspace{0.5cm} \leq \alpha + \alpha \beta\Big\{-1+\mcal{O}(\beta) + \frac{1}{2} \alpha^2 + \mcal{O}(\beta) \delta \alpha^{-1} + \mcal{O}(\beta)\delta^2 \alpha^{-2} + \frac{2\delta}{\beta}\alpha^{-1} \Big\} \\
& \hspace{0.5cm} \leq \alpha + \alpha\beta \Big \{-1 + \mcal{O}(\beta) + \mcal{O}(\beta^\epsilon) + \mcal{O}(\beta^{1+\frac{\epsilon}{2}}) + \mcal{O}(\beta^{3+\epsilon}) + \mcal{O}(\beta^{\frac{\epsilon}{2}}) \Big \}.
\end{aligned}
\end{align*}
This yields the desired result that $\mcal{D}(\mcal{V}(n+1)) < \alpha$ for sufficiently small $\beta$'s.
\end{proof}

\vspace{0.2cm}

Next, we study time-evolution of the following quantity:
\[  \snorm{V_i V_j^\dag - \tilde{V}_i\tilde{V}_j^\dag }_\F^2. \]
By \eqref{D-16-0}, we get the following analogue of \eqref{D-12-1-3}:
\begin{equation} \label{D-17-1}
\snorm{V_i V_j^\dag - \tilde{V}_i\tilde{V}_j^\dag }_\F^2 (n+1) - \snorm{V_i V_j^\dag - \tilde{V}_i \tilde{V}_j^\dag}_\F^2 (n) = - \sum_{m=1}^\infty \frac{\beta^m}{m!} \widehat{\mcal{T}}_m,
\end{equation}
where, for $m\geq 1$,
\[\widehat{\mcal{T}}_m = \sum_{k=0}^m \frac{m!}{k!(m-k)!} \bigg\{\tr{ \mcal{I}_k^+ V_i V_j^\dag \mcal{J}_{m-k}^- \tilde{V}_j \tilde{V}_i^\dag} + \tr{\mcal{J}_{m-k}^+ V_j V_i^\dag \mcal{I}_{k}^- \tilde{V}_i \tilde{V}_j^\dag} \bigg\}.\]

\noindent First, regarding \eqref{D-16-2} and the estimate \eqref{D-13} for ${\mcal{I}}^\pm_k$ and ${\mcal{J}}^\pm_k$, we get the following analogue of Lemma \ref{L6.2} for $\widehat{\mcal{T}}_m$:
\begin{equation} \label{D-17-2}
\big|\widehat{\mcal{T}}_m \big| \leq (2^{2m-1} + m\cdot 2^m) d(\mcal{V}, \tilde{\mcal{V}})^2,\quad m\geq 2.
\end{equation}

\noindent Next, we consider
\begin{align*}
\begin{aligned}
& \mcal{T}_1^{\mcal{V}} := \mathrm{tr}\left((\vDelta_i - \tvDelta_i) V_i V_j^\dag \tilde{V}_j \tilde{V}_i^\dag - V_i V_j^\dag (\vDelta_j - \tvDelta_j) \tilde{V}_j \tilde{V}_i^\dag \right.\\
& \hspace{1.2cm} \left. + (\vDelta_j - \tvDelta_j) V_j V_i^\dag \tilde{V}_i \tilde{V}_j^\dag - V_j V_i^\dag (\vDelta_i - \tvDelta_i) \tilde{V}_i \tilde{V}_j^\dag \right) \\
& \hspace{0.2cm} = \mathrm{tr} \left((\vDelta_i - \tvDelta_i)(V_i V_j^\dag \tilde{V}_j \tilde{V}_i^\dag - \tilde{V}_i \tilde{V}_j^\dag V_j V_i^\dag) + (\vDelta_j - \tvDelta_j) (V_j V_i^\dag \tilde{V}_i \tilde{V}_j^\dag -  \tilde{V}_j \tilde{V}_i^\dag V_i V_j^\dag) \right).
\end{aligned}
\end{align*}
By Lemma \ref{L5.1}, Lemma \ref{L5.2} and Lemma \ref{L6.1}, one has
\begin{align*}
\begin{aligned}
& \Big\|(V_i V_j^\dag - \tilde{V}_i \tilde{V}_j^\dag) - (U_i U_j^\dag - \tilde{U}_i \tilde{U}_j^\dag)\Big\|_\F \\
& \hspace{0.5cm} = \norm{(V_i V_j^\dag - \tilde{V}_i \tilde{V}_j^\dag) -  \exp\left(\frac{\mathrm{i}H_i h}{2} \right) (V_i V_j^\dag - \tilde{V}_i \tilde{V}_j^\dag) \exp\left(\frac{-\mathrm{i}H_j h}{2} \right)}_\F \\
& \hspace{0.5cm} \leq \left( \norm{I - \exp\left(\frac{\mathrm{i}H_i h}{2} \right)}_\F + \norm{I- \exp\left(\frac{-\mathrm{i}H_j h}{2} \right)}_\F \right) \norm{V_i V_j^\dag - \tilde{V}_i \tilde{V}_j^\dag}_\F \\
& \hspace{0.5cm} \leq h\cdot \mcal{D}(\mcal{H}) \cdot d(\mcal{V},\tilde{\mcal{V}}),
\end{aligned}
\end{align*}
for all $i,j\in[N]$. Thus, we get
\begin{align*}
\begin{aligned}
& \snorm{(\vDelta_i - \tvDelta_i) - (\Delta_i - \tilde{\Delta}_i)}_\F \\
& \hspace{0.5cm} = \frac{1}{2N} \sum_{k=1}^N \left\{ \norm{ (V_k V_i^\dag - \tilde{V}_k \tilde{V}_i^\dag) - (U_k U_i^\dag - \tilde{U}_k \tilde{U}_i^\dag)}_\F \right. \\
& \hspace{2cm} \left. + \norm{ (V_i V_k^\dag - \tilde{V}_i \tilde{V}_k^\dag) - (U_i U_k^\dag - \tilde{U}_i \tilde{U}_k^\dag)}_\F \right\} \\
& \hspace{0.5cm} \leq \frac{1}{2N} \cdot N \cdot d(\mcal{V}, \tilde{\mcal{V}}) \cdot h\cdot \mcal{D}(\mcal{H}) \cdot 2 = h \mcal{D}(\mcal{H}) d(\mcal{V}, \tilde{\mcal{V}}).
\end{aligned}
\end{align*}
By Lemma \ref{L5.3}, one has
\[\snorm{V_i V_j^\dag \tilde{V}_j \tilde{V}_i^\dag - \tilde{V}_i \tilde{V}_j^\dag V_j V_i^\dag}_\F \leq \snorm{V_i V_j^\dag - \tilde{V}_i \tilde{V}_j^\dag}_\F + \snorm{\tilde{V}_j \tilde{V}_i^\dag - V_j V_i^\dag}_\F \leq 2 d(\mcal{V},\tilde{\mcal{V}}),\]
for $i,j\in[N]$. We combine the above inequalities by Lemma \ref{L5.1} and Lemma \ref{L5.2} to get
\begin{align}
\begin{aligned} \label{D-17-3}
\big| \mcal{T}_1^{\mcal{V}} - \widehat{\mcal{T}}_1 \big| & \leq \snorm{(\vDelta_i - \tvDelta_i) - (\Delta_i - \tilde{\Delta}_i)}_\F \snorm{V_i V_j^\dag \tilde{V}_j \tilde{V}_i^\dag - \tilde{V}_i \tilde{V}_j^\dag V_j V_i^\dag}_\F \\
& \quad + \snorm{(\vDelta_j - \tvDelta_j) - (\Delta_j - \tilde{\Delta}_j)}_\F \snorm{V_j V_i^\dag \tilde{V}_i \tilde{V}_j^\dag -  \tilde{V}_j \tilde{V}_i^\dag V_i V_j^\dag}_\F \\
& \leq 4h \mcal{D}(\mcal{H}) d(\mcal{V},\tilde{\mcal{V}})^2.
\end{aligned}
\end{align}
Therefore, by \eqref{D-13-1} and \eqref{D-17-3}, we have
\begin{equation} \label{D-17-4}
\widehat{\mcal{T}}_1 \geq \mcal{T}_1^{\mcal{V}} - \big| \mcal{T}_1^{\mcal{V}} - \widehat{\mcal{T}}_1 \big| \geq 2\snorm{V_i V_j^\dag - \tilde{V}_i \tilde{V}_j^\dag}_\F^2 - (4\mcal{D}(\mathcal{V}) + 2\mcal{D}(\tilde{\mathcal{V}}) + 4h \mcal{D}(\mcal{H})) d(\mathcal{V},\tilde{\mathcal{V}})^2.
\end{equation}
Suppose we have a prior estimate that $\mcal{D}(\mcal{V}), \mcal{D}(\tilde{\mcal{V}}) < \alpha$ as long as solution exists. Then, we use \eqref{D-17-1} and the estimates of $\widehat{\mcal{T}}_m$'s \eqref{D-17-2}, \eqref{D-17-4} to get the following analogue of \eqref{D-14-1}:
\begin{align}
\begin{aligned} \label{D-17-5}
\snorm{V_i V_j^\dag - \tilde{V}_i\tilde{V}_j^\dag }_\F^2 (n+1) &= \snorm{V_i V_j^\dag - \tilde{V}_i \tilde{V}_j^\dag}_\F^2 (n) - \sum_{m=1}^\infty \frac{\beta^m}{m!} \widehat{\mcal{T}}_m \\
&\leq (1-2\beta) \snorm{V_i V_j^\dag - \tilde{V}_i \tilde{V}_j^\dag}_\F^2 (n) \\
& \quad + \left(6\alpha \beta + 4 h \mcal{D}(\mcal{H}) + \sum_{m=2}^\infty \frac{\beta^m}{m!}(2^{2m-1} + m\cdot 2^m) \right) d(\mathcal{V},\tilde{\mathcal
{V}})^2 \\
&\leq \left( 1- \beta \left(6- 2 e^{2\beta} - \frac{e^{4\beta}-1}{2\beta} - 6\alpha - \frac{4\mcal{D}(\mcal{H})}{\kappa} \right) \right) d(\mathcal{V}, \tilde{\mathcal{V}})^2 \\
&= \left( 1- 6\beta \left(M(\beta) - \alpha - \frac{3}{2}\cdot \frac{\mcal{D}(\mcal{H})}{\kappa} \right) \right) d(\mathcal{V}, \tilde{\mathcal{V}})^2 .
\end{aligned}
\end{align}
\subsubsection{Asymptotic state-locking} Previous estimates \eqref{D-17-5} for the variation of \\ $d(\mcal{U}(n),\tilde{\mcal{U}}(n))$ is similar to \eqref{D-14-1}. Thus, a solution $\{\mcal{U}(n)\}$ of the discrete Lohe matrix model C \eqref{D-16} is expected to enrich the asymptotic aggregation of the relative positions under sufficient conditions on the initial data $\mcal{U}^0$ and the model parameters $\kappa$, $\beta$, and $\mcal{H}$.

\begin{theorem} \label{T6.3}
Suppose system parameters and initial data satisfy 
\begin{align*}
& 0 < \beta \ll 1, \quad 0< \frac{\mcal{D}(\mcal{H})}{\kappa} \lesssim \beta^{1+\epsilon} , \quad \sum_{k=1}^N H_k = O , \mathcal{U}^0\in \overline{\mcal{B}(\alpha)},\quad 0 < \alpha < M(\beta) - \frac{2 \mcal{D}(\mcal{H})}{\kappa},
\end{align*}
and let $\{\mathcal{U}(n)\}$ and $\{\tilde{\mathcal{U}}(n)\}$ be solutions to \eqref{D-16} corresponding to the initial data $\mathcal{U}^0$ and $\tilde{\mathcal{U}}^0$ respectively,
and $\,\mathcal{U}(n), \tilde{\mathcal{U}}(n)\in\overline{\mcal{B}(\alpha)}$ for all $n\geq 0$. Then the following assertions hold:
\begin{enumerate}[label=(\roman*)]
\item The relative positions synchronize exponentially fast:~there exists a constant $0<C<1$ such that 
\[d(\mathcal{U}(n), \tilde{\mathcal{U}}(n)) \leq C^n d(\mathcal{U}^0,\tilde{\mathcal{U}}^0),\quad n\geq 0.\]
\item There exists a unitary matrix $L^\infty \in \mathbf{U}(d)$ independent of $i$ such that
\[\lim_{n\to\infty} U_i(n) \tilde{U}_i^\dag(n) = L^\infty\]
\item The discrete-time Lohe flow $\{\mathcal{U}(n)\}_{n\geq 0}$ achieves asymptotic state-locking, i.e., \[\lim_{n\to\infty} U_i(n) U_j^\dag (n)\]
converges exponentially fast for all $i, j \in [N]$.
\end{enumerate}
\end{theorem}

\begin{proof} According to \eqref{D-16-1}, the given condition 
\[ \mcal{U}(n), \tilde{\mcal{U}}(n)\in \overline{\mcal{B}(\alpha)} \quad \mbox{for all $n\geq 0$} \]
implies 
\[ \mcal{V}(n), \tilde{\mcal{V}}(n)\in \overline{\mcal{B}\left(\alpha + \frac{\mcal{D}(\mcal{H})}{2 \kappa} \right)} \quad \mbox{for all $n\geq 0$}. \]
Regarding \eqref{D-16-2} and estimate \eqref{D-17-5}, the proof can be done by the way we did in Theorem \ref{T6.1} and Theorem \ref{T6.2}.
\end{proof}

\section{Conclusion} \label{sec:7}
\setcounter{equation}{0}
In this paper, we have provided several discrete models corresponding to the discrete counterparts of the Lohe sphere model and Lohe matrix model on the unitary group. Aforementioned models are the high-dimensional generalizations of the Kuramoto model for synchronization. The Lohe sphere model describes the aggregate modeling on the unit sphere, whereas the Lohe matrix model describes the aggregate modeling on a unitary group. Recently, emergent dynamics and phase transition phenomena for  these continuous models have been extensively studied in literature from control theory and statistical physics communities. When one tries to simulate these models, we are forced to discretize these continuous models using suitable discretization algorithms.

\vspace{0.2cm}

For the discretization of the Lohe sphere model, we construct a scheme consisting of two steps. In the first step, we use the first-order forward Euler scheme to get the intermediate state which is not on the unit sphere. In the second step, we project the intermediate state to the unit sphere for the state in next time step. Since projection operator is contractive, all the emergent properties of intermediate states carry over to the projected states.

\vspace{0.2cm}

On the other hand, for the Lohe matrix model, we proposed three discrete models derieved from the Lie group exponential map. The first one is based on the exponential map on the underlying Lie group. This corresponds to the manifold version of the first-order Euler scheme free of projections. The other two discrete models are motivated by operator splitting theory. For a homogeneous flow with the same free flow, we show that complete state aggregation occurs exponentially fast under a suitable framework in terms of system parameters and initial data. On the other hand, for a heterogenous flow with different free flows, we can not prove emergent dynamics yet. The second and third discrete models are motivated by the Lie-Trotter splitting scheme and Strang splitting scheme.  For zero free flow, all three discrete models coincide and exhibit exponential aggregation, whereas for the last two discrete models based on operator splitting exhibit the orbital stability and the state-locking states.

\end{document}